\documentclass[a4paper,UKenglish,cleveref, autoref]{lipics-v2021}
\usepackage{amsmath}
\usepackage{stmaryrd}
\usepackage{esint}
\usepackage{mdframed}
\usepackage{quiver}
\SetSymbolFont{stmry}{bold}{U}{stmry}{m}{n}

\usepackage{makecell}
\usepackage{preamble}
\usepackage{graphicx}
\usepackage{bm} 
\usepackage{bbm} 
\usepackage{newmacros}

\pdfoutput=1 
\hideLIPIcs  

\title{Directed type theory, with a twist} 

\author{Fernando Chu}{Utrecht University, Netherlands}{johnqpublic@dummyuni.org}{https://orcid.org/0009-0005-9125-769X}{}

\author{Paige Randall North}{Utrecht University, Netherlands}{p.r.north@uu.nl}{https://orcid.org/0000-0001-7876-0956}{}

\authorrunning{F. Chu and P.\,R. North} 

\Copyright{Fernando Chu and Paige Randall North} 



\category{} 

\relatedversion{} 

\acknowledgements{The authors thank Kenji Maillard, Andrea Laretto, Christian Merten, Rob Schellingerhout, Ulrik Buchholtz, and \'El\'eonore Mangel for insightful discussions about this work.}

\EventEditors{John Q. Open and Joan R. Access}
\EventNoEds{2}
\EventLongTitle{42nd Conference on Very Important Topics (CVIT 2016)}
\EventShortTitle{CVIT 2016}
\EventAcronym{CVIT}
\EventYear{2016}
\EventDate{December 24--27, 2016}
\EventLocation{Little Whinging, United Kingdom}
\EventLogo{}
\SeriesVolume{42}
\ArticleNo{23}

\begin{document}

\nolinenumbers 
\maketitle

\begin{abstract}
  In recent years, Homotopy Type Theory (HoTT) has had great success both as a foundation of mathematics and as internal language to reason about $\infty$-groupoids (a.k.a. spaces).
  However, in many areas of mathematics and computer science, it is often the case that it is categories, not groupoids, which are the more important structures to consider.
  For this reason, multiple directed type theories have been proposed, i.e., theories whose semantics are based on categories.
  In this paper, we present a new such type theory, Twisted Type Theory (TTT).
  It features a novel ``twisting'' operation on types: given a type that depends both contravariantly and covariantly on some variables, its twist is a new type that depends only covariantly on the same variables.
  To provide the semantics of this operation, we introduce the notion of dependent 2-sided fibrations (D2SFibs), which generalize Street's notion of 2-sided fibrations.
  We develop the basic theory of D2SFibs, as well as characterize them through a straightening-unstraightening theorem.
  With these results in hand, we introduce a new elimination rule for Hom-types.
  We argue that our syntax and semantics satisfy key features that allow reasoning in a HoTT-like style, which allows us to mimic the proof techniques of that setting.
  We end the paper by exemplifying this, and use TTT to reason about categories, giving a syntactic proof of Yoneda's lemma.
\end{abstract}

\section{Introduction}
In a seminal paper \cite{hofmann_groupoid_1998}, Hofmann and Streicher constructed the \emph{groupoid model}: a semantic interpretation of Martin-Löf Type Theory (MLTT) \cite{martin-lof_intuitionistic_1975} where contexts are interpreted as groupoids and types are interpreted as indexed groupoids.
Years later, Voevodsky et al.\ \cite{kapulkin_simplicial_2021} generalized this result by constructing the \emph{$\infty$-groupoid model}.
Further, they observed that this model satisfies the so-called \emph{univalence axiom}, which roughly states that equivalent types are equal.
These results marked the beginning of Homotopy Type Theory (HoTT) \cite{the_univalent_foundations_program_homotopy_2013}.

After multiple years of development, HoTT has proven to be both (i) a successful foundation for developing mathematics and (ii) a useful synthetic language to reason about $\infty$-groupoids (a.k.a. spaces).
As an example of (i), in HoTT one can use the Structure Identity Principle (SIP) to transport properties between equivalent algebraic structures \cite{coquand_isomorphism_2013,ahrens_higher_2020}; this gives a rigorous foundation to the often used practice in mathematics of equating equivalent structures.
As an example of (ii), one can encode CW-complexes as Higher Inductive Types (HITs) and compute their homotopy groups synthetically \cite{brunerie_james_2019,hutchison_pi_n_2013}.

However, in many areas of mathematics and computer science, it is often the case that it is categories, not groupoids, which are the more important structures to consider.
Hence, multiple researchers have proposed \emph{directed type theories} \cite{ahrensBicategoricalTypeTheory2023,gratzer_directed_2024,larettoDiDirectedFirstOrder2026,licata_2-dimensional_2011,nasuInternalLogicVirtual2024,neumannSynthetic1CategoriesDirected2025,newFormalLogicFormal2023,north_towards_2019,andreas_towards_2015,riehl_type_2017}.
That is, type theories with a notion of non-symmetric ``directed equality'', whose semantics are based on categories.
The hope is then, in analogy to HoTT, that these theories would be (i) via a directed version of the SIP, an even more powerful foundation of mathematics and (ii) a useful synthetic language to reason about categories.

In this paper, we present a new directed type theory, which we call \emph{Twisted Type Theory (TTT)}.
One of our overarching design principles was to follow very closely Hofmann and Streicher's groupoid model of MLTT.
Therefore, we build upon North's category model \cite{north_towards_2019}, which generalizes it.
Indeed, in the category model, contexts are interpreted as categories, while types are interpreted as indexed categories.

We focus on generalizing the groupoid model as it is simpler than its $\infty$-dimensional counterpart, yet still retains enough features to validate a restricted version of points (i) and (ii) above.
Following closely the groupoid model will therefore allow us to replicate the syntactic techniques of HoTT, while maintaining some level of simplicity in our semantic considerations.
We leave it to future work to explore the higher dimensional models of TTT.

More specifically, TTT was designed to satisfy the following three desiderata, each of which is a categorical adaptation of a feature the groupoid model satisfies.

\begin{enumerate}
  \item \emph{Types as categories}.
    Being a type in the empty context should semantically correspond to being a category.
    This has the benefit of reducing the problem of whether a new construction is a category to a typechecking exercise.
    This is in contrast to other work, such as Simplicial Type Theory (STT) \cite{riehl_type_2017}, that extends the universe of types to include additional structures that are not the object of interest, i.e.\ they have types that are not categories.
    Categories are then recovered as types satisfying some additional conditions, and therefore one has to provide a proof that a newly constructed type is a category.
    While there are some results that ensure that certain constructions in STT preserve categories, one is usually also interested in obtaining categories by applying constructions on the more general base types.
    For example, discrete types in STT are defined as types satisfying a certain discreteness condition. Ad-hoc arguments are then used to show that these satisfy the so-called Segal condition, see Section 7 in op.\ cit.

  \item \emph{Emergent categorical structure}.
    There should be a type of directed equalities, i.e.\ a $\Hom$-type, given by formation, introduction, elimination and computation rules.
    This condition eases a generalization of the type theory to higher dimensions.
    Indeed, having $\Hom$-types allow encoding the $n$-cells of a type as the $n$-th iteration of the $\Hom$-type.
    Further, from an appropriate $\Hom$-elimination rule, the coherences will follow; this is what happens with the $\Id$-elimination rule in HoTT, which gives types the structure of $\infty$-groupoids.
    In contrast, other approaches would have to significantly change the type theory to accommodate the $n$-cells and their coherences.
    For example, a higher dimensional analogue of the 2-dimensional type theory of Licata \& Harper \cite{licata_2-dimensional_2011}, which does not have $\Hom$-types, would have to manually specify an $n$-categorical structure to the contexts. Similarly, STT would have to add a calculus of new shapes \cite[p. 152]{riehl_type_2017}.

  \item \emph{Trivial directed equalities through directed path objects}. We require that the $\Hom$-introduction rule corresponds to the factorization $A \to A^\to \to A\times A$ of the diagonal map $\delta:A \to A\times A$ through the arrow category $A^\to$ of a category $A$ (see Section \ref{sec:d2sfibs}).
    This is a generalization of the fact that, in a model category, path objects, i.e.\ appropriate factorizations $A \to A^I \to A \times A$ of the diagonal map $\delta$, can be used to give an interpretation of the formation and introduction rules of $\Id$-types \cite{awodey_homotopy_2009}.
    This technical condition ensures that the HoTT notion of homotopies between functions, i.e.\ functions $\Pi_{x:X}f(x) = g(x)$, can be readily generalized to the directed case.
    That is, it allows for reasoning about natural transformations (see Subsection \ref{subsec:nats}).
    This had previously been a difficulty in some directed type theories, see e.g.\ Neumann \& Altenkirch \cite[Section 6]{neumannSynthetic1CategoriesDirected2025}.
\end{enumerate}

The work of North \cite{north_towards_2019} already satisfied the first two desiderata.
Our contribution is then satisfying also the third.
We achieve this by introducing a novel ``twisting'' operation on types: given a type that depends both contravariantly and covariantly on some variables, its twist is a new type that depends only covariantly on the same variables.
Crucially, the interpretation of the twist of the $\Hom$-type of a type $A$ gives the arrow category $A^\to$ of $A$.
This enables us to state a $\Hom$-introduction rule, syntactically analogous to the $\Id$-introduction rule, whose interpretation satisfies desideratum (3).

More generally, the interpretation of the twisting of a type makes use of our notion of \emph{dependent 2-sided fibrations (D2SFibs)}, a generalization of Street's notion of a 2-sided fibration \cite{streetFibrationsYonedasLemma1974}.
Indeed, the twist of a type corresponds then to the unstraightening part of a straightening-unstraightening characterization we give of D2SFibs.

The theory TTT can then be summarized as an extension of North's type theory \cite{north_towards_2019} featuring twisted types and modified $\Hom$-types.
We will also make use of other features already well known in the literature, such as restricted $\Pi$-types, opposite contexts, discrete universes, among others \cite{licata_2-dimensional_2011,neumannSynthetic1CategoriesDirected2025}.




\subparagraph*{Outline and main contributions.}
\begin{itemize}
  \item Section \ref{sec:review} reviews the categorical semantics of MLTT and North's category model \cite{north_towards_2019}, as well as introduces the displayed category model of MLTT.
  \item Section \ref{sec:d2sfibs} introduces D2SFibs and develops the basic theory of these, which is used to give the semantics of the twist of a type.
  \item Section \ref{sec:TTT} introduces TTT, which presents novel rules for $\Hom$-types and adapts some additional features previously considered in the literature to our context.
  \item Section \ref{sec:applications} presents some applications of TTT, including a proof of Yoneda's lemma.
  \item Section \ref{sec:future-work} concludes the paper and sketches future work directions.
\end{itemize}

\subparagraph*{Related work.}

The first published work on directed type theory was presented by Licata \& Harper \cite{licata_2-dimensional_2011}.
Their theory's main feature is that the syntax had a primitive notion of transformations between parallel substitutions of contexts, thereby giving a 2-categorical structure to the category of contexts.
Since they did not introduce $\Hom$-types, their work does not satisfy our desiderata (2) nor (3).
Later work by North \cite{north_towards_2019} introduced $\Hom$-types for the first time, but did so by restricting the $\Hom$-introduction rule to the groupoid core of a type, hence their theory does not meet our desiderata (3).
More recently, Riehl \& Shulman \cite{riehl_type_2017} extended MLTT with a directed interval and other shapes, resulting in a powerful calculus reminiscent of Cubical Type Theory \cite{cohenCubicalTypeTheory2018}.
In order to do this, their semantic interpretation of a type is a bisimplicial set, and hence only some types correspond to categories; therefore our desiderata (1) and (2) are not satisfied.

Further work has extended and refined each of these three works in different directions: Ahrens et al.\ \cite{ahrensBicategoricalTypeTheory2023} generalize Licata \& Harper's work by giving the category of contexts a bicategorical structure, Neumann \& Altenkirch \cite{neumannSynthetic1CategoriesDirected2025} extend North's type theory by also adding groupoid cores of contexts and hence obtaining new rules for $\Hom$-types, and Gratzer et al.\ \cite{gratzer_directed_2024} add modalities to Riehl \& Shulman's type theory, greatly increasing the expressivity of the theory.
Nonetheless, these new works satisfy the same desiderata as the work their building upon, and so none satisfy all three of our criteria.
Table \ref{table:work} summarizes this discussion.

\begin{table}[h]
  \begin{tabular}{lllll}
    \cline{2-5}

    \multicolumn{1}{c|}{\multirow{2}{*}{}} & \multicolumn{1}{c|}{\multirow{2}{*}{Works}} & \multicolumn{3}{c|}{Desiderata}                                                   \\ \cline{3-5}
    \multicolumn{1}{c|}{}                   & \multicolumn{1}{c|}{}                   & \multicolumn{1}{c|}{(1)} & \multicolumn{1}{c|}{(2)} & \multicolumn{1}{c|}{(3)} \\ \hline
    \multicolumn{1}{|l|}{\makecell[l]{Directed type theories with two-dimensional\\ categories of contexts}}                  & \multicolumn{1}{l|}{\makecell[l]{Licata \& Harper \cite{licata_2-dimensional_2011}, \\ Ahrens et al.\ \cite{ahrensBicategoricalTypeTheory2023}}}                  & \multicolumn{1}{l|}{Yes} & \multicolumn{1}{l|}{No} & \multicolumn{1}{l|}{No} \\ \hline
    \multicolumn{1}{|l|}{\makecell[l]{Directed type theories with operations\\ involving groupoids}}                  & \multicolumn{1}{l|}{\makecell[l]{North \cite{north_towards_2019}, Neumann \\ \& Altenkirch \cite{neumannSynthetic1CategoriesDirected2025}}}                  & \multicolumn{1}{l|}{Yes} & \multicolumn{1}{l|}{Yes} & \multicolumn{1}{l|}{No} \\ \hline
    \multicolumn{1}{|l|}{Directed type theories with directed intervals}                  & \multicolumn{1}{l|}{\makecell[l]{Riehl \& Shulman \cite{riehl_type_2017}, \\Gratzer et al.\ \cite{gratzer_directed_2024}}}                  & \multicolumn{1}{l|}{No} & \multicolumn{1}{l|}{No} & \multicolumn{1}{l|}{Yes} \\ \hline
  \end{tabular}
  \caption{Summary of related work and desiderata satisfied.}
  \label{table:work}
\end{table}



A different line of work on directed type theory comes from theories that do not extend MLTT.
These include the simply typed theory of Laretto et al.\ \cite{larettoDiDirectedFirstOrder2026}, which gives a (co)end calculus for profunctors; as well as the work of New \& Licata \cite{newFormalLogicFormal2023} and Nasu \cite{nasuInternalLogicVirtual2024}, which introduce type theories whose semantics are based on virtual equipments and virtual double categories, respectively.
While these theories informed the present work, their lack of dependent types prevents us from directly comparing them with TTT.

The work of Nuyts \cite{andreas_towards_2015} is also relevant, as it introduces a modal directed type theory expanding on the work of Licata \& Harper \cite{licata_2-dimensional_2011}.
However, it does not provide semantics for their calculus, and so does not satisfy any of our desiderata.

\section{Categorical models of MLTT}\label{sec:review}
For the rest of this paper, all our (op)fibrations are cloven.
We begin by recalling the definition of comprehension category.

\begin{definition}[Comprehension categories \cite{jacobsComprehensionCategoriesSemantics1993}]
  Let $\mathcal{C}$ be a  category. A \emph{(split) comprehension category (CC) over} $\mathcal{C}$ is a (split) fibration $p:\mathcal{T}\to \mathcal{C}$ equipped with a functor $\chi: \mathcal{T} \to \mathcal{C}^\to$, the \emph{comprehension}, such that it preserves cartesian morphisms and such that the following diagram commutes.
  \[
    \begin{tikzcd}[sep=scriptsize]
      \mathcal{T} && \mathcal{C}^\to \\
      & \mathcal{C}
      \arrow["\chi", from=1-1, to=1-3]
      \arrow["p"', from=1-1, to=2-2]
      \arrow["\cod", from=1-3, to=2-2]
  \end{tikzcd}\]
  We write $(\mathcal{C}, \mathcal{T}, \chi)$ for such a CC. We say the CC is \emph{full} if $\chi$ is full and faithful.
\end{definition}

Jacobs \cite{jacobsComprehensionCategoriesSemantics1993} showed that full split CC's provide semantics for MLTT.
Briefly, the interpretation of the syntax is as follows:
\begin{itemize} 
  \item A context $\Gamma$ is interpreted as an object $\Interp{\Gamma} \in \mathcal{C}$.
  \item A type $\Gamma \vdash A\,\Ty$ is interpreted as an object $\Interp{A}\in \mathcal{T}_\Interp{\Gamma}$, i.e.\ the fibre in $\mathcal{T}$ over $\Interp{\Gamma}$.
  \item The context extension $\Gamma, a:A$ is interpreted as the domain of the comprehension of $A$, and so $\Interp{\Gamma, x:A} \defeq \mathsf{dom}(\chi(\Interp{A}))$.
  \item A term $\Gamma \vdash a :A$ is interpreted as a section $\Interp{a}: \Interp{\Gamma} \to \Interp{\Gamma.A}$ of $\chi(\Interp{A})$.
\end{itemize}

We now will give two categorical models of MLTT, i.e.\ two CC's over $\Cat$. To simplify the following discussion, we introduce the following notations.
\begin{notation}
  Let $F : A \to \Cat$ (resp. $G : B^\op \to \Cat$) be a functor. We write $A \ctxe F$ (resp. $B \ctxe^- G$) for the covariant (resp. contravariant) Grothendieck construction of $F$ (resp. $G$). We write $\pi:A\ctxe F \to A$ (resp. $\pi:B \ctxe^- G \to B)$ for its canonical projection.
\end{notation}

\subsection{The category model of MLTT}

The category model of MLTT in terms of a CC is given by North \cite[Construction 3.3]{north_towards_2019}, which we now recall.
\begin{definition}[The category model]\label{def:catmodel}
  Consider the functor $[-,\Cat]_\textnormal{lax}:\Cat^\op \to \Cat$
  that maps a category $A$ to the category of strict functors from $A$ to $\Cat$ and lax natural transformations. Write $\ICat_\textnormal{lax}$ for the category $\Cat \ctxe^- [-,\Cat]_\textnormal{lax}$ of indexed categories and lax natural transformations.

  The \emph{category model} is the full split CC over $\Cat$ given by the canonical projection of the $\pi:\ICat_\textnormal{lax} \to \Cat$, together with the comprehension functor $\ctxe^*{} : \ICat_\textnormal{lax} \to \Cat^\to$ given by mapping an object $(A,B)$ of $\ICat_\textnormal{lax}$ to $A \ctxe B \xrightarrow{\pi} A$ in $\Cat^\to$.
\end{definition}

It follows from the above discussion that, in the category model, a context is interpreted as a category $\Interp{\Gamma}$, while a type $\Gamma \vdash A\,\Ty$ interpreted as an indexed category $\Interp{A}:\Interp{\Gamma} \to \Cat$.

\begin{notation}
  From now on we write $\cxe$ for the context extension operation in MLTT, i.e.\ we write $\Gamma \cxe a : A$ for $\Gamma, a : A$.
  For the category model, we then have $\Interp{\Gamma \cxe a : A} \defeq \Interp{\Gamma} \ctxe \Interp{A}$.
  This notation will be useful to disambiguate from additional context extension operations we will introduce later.
\end{notation}

\begin{proposition}[\cite{licata_2-dimensional_2011,north_towards_2019}]\label{prop:Sigma}
  The category model has strong $\Sigma$-types and, if a countable hierarchy of inaccessible cardinals is assumed, universes $\UU_i$.
  Further, it also validates the following rules about opposite types and $\Hom$-types.
  {
    \small
    \begin{mathparpagebreakable}
      \inferrule*[right=$\op$\textnormal{\textsc{-Form}}]{
        \IsTy{A}
      }{
        \IsTy{A^\op}
      } \and
      \inferrule*[right=$\Hom$\textnormal{\textsc{-Form}}]{
        \IsTy{A}
      }{
        \IsTy[\Gamma \cxe a:A^\op \cxe b:A]{\Hom_{A}(a,b)}
      } \and
      \inferrule*[right=$\op$\textnormal{\textsc{-Inv}}]{
        \IsTy{A}
      }{
        \EqTy{(A^\op)^\op}{A}
      } \and
      \inferrule*[right=$\Hom$\textnormal{\textsc{-}}$\op$]{
        \IsTy{A}
      }{
        \IsTy[\Gamma \cxe a:A^\op \cxe b:A]{\Hom_{A}(a,b)^\op \jeq \Hom_{A}(a,b)}
      }
    \end{mathparpagebreakable}
  }
\end{proposition}
\begin{proof}[Proof sketch]
  Let $\Gamma \vdash A\,\Ty$ and $\Gamma \cxe a : A \vdash B\,\Ty$ be types.
  The interpretation of $\Gamma \vdash \sum_AB\,\Ty$ is given by setting $\Interp{\sum_AB}(\gamma)$ to be the category that has as objects $(a,b) \in \coprod_{a\in A(\gamma)}B(a,\gamma)$, while a morphism $(a,b) \to (a',b')$ is a pair $(\alpha,\beta)$ with $\alpha:a \to a'$ in $A(\gamma)$ and $\beta:B(\id_\gamma, \alpha)b \to b'$.
  The interpretation of opposite types is given by postcomposition with the opposite functor $\op : \Cat \to \Cat$. Finally, the interpretation of $\Hom$-types is given by $\Interp{\Hom_A}(\gamma,x,y) \defeq \hom_{A(\gamma)}(x,y)$.
\end{proof}

\begin{remark}\label{remark:univs}
  For the concerned reader, the specific universes the category model validates (under the assumption above) are Coquand universes \cite{coquandPresheafModelType2013,gratzer_multimodal_2021}.
  However, for convenience, we will use Russell universes with implicit indices, understanding that these should be elaborated to the Coquand ones.
\end{remark}

\begin{notation}
  For a functor $F: C \to \Cat$, we will write $F^\op$ for $\op \circ F$. With this notation, for a type $A$, we have $\Interp{A^\op} \defeq \Interp{A}^\op$.
\end{notation}

The following well-known result, due to Grothendieck \cite{grothendieckRevetementsEtalesGroupe2003}, gives us a different perspective on this comprehension category.
\begin{proposition}[Straightening-unstraightening for opfibrations]\label{prop:straight-icat}
  Write $\Opfib^\textup{split}_\textup{lax}$ for the full fibred subcategory of $\Cat^\to$ spanned by the split opfibrations.
  The comprehension functor $\ctxe^*{} : \ICat_\textnormal{lax} \to \Cat^\to$ restricts to a fibred equivalence of categories $\ctxe^*\, : \ICat_\textnormal{lax} \xrightarrow{\simeq} \Opfib^\textup{split}_\textup{lax}$ over $\Cat$.
\end{proposition}

Note that the fibration $\cod:\Opfib^{\text{split}}_\text{lax}\to \Cat$ is not split, and hence cannot be directly used to interpret MLTT.
Hence, we will carefully apply Proposition \ref{prop:straight-icat}, ensuring that we are not breaking the splitness of our constructions.

\subsection{The displayed category model of MLTT}
As we mentioned in the introduction, for a category $A$, we are interested in a factorization $A \xrightarrow{\iota} A^\to \xrightarrow{\langle\mathsf{dom},\mathsf{cod}\rangle} A \times A$ of the diagonal map $\delta :A \to A \times A$.
More specifically, we would like $\iota$ to be defined by $\iota(a) \defeq \mathsf{id}_a$.
In Hofmann \& Streicher's groupoid model \cite{hofmann_groupoid_1998}, the interpretation of the term $a:A \vdash \Refl_a : \mathsf{Id}(a,a)$ carries precisely this data, as the canonical projection $\pi:\Interp{a:A,b:A,\Id(a,b)} \to \Interp{a:A,b:A}$ is isomorphic to the functor $\langle\mathsf{dom},\mathsf{cod}\rangle:\Interp{A}^\to \to \Interp{A} \times \Interp{A}$.

Following this pattern, for the category model case, we are led to consider a type $a:A \cxe b:A \vdash H(a,b) \,\Ty$, such that $\pi:\Interp{a:A \cxe b : A \cxe H(a,b)} \to \Interp{a:A\cxe b : B}$ is isomorphic to $\langle\mathsf{dom},\mathsf{cod}\rangle:A^\to \to A \times A$.
However, this approach cannot work: by Proposition \ref{prop:straight-icat}, the projection $\pi$ is always an opfibration, but $\langle\mathsf{dom},\mathsf{cod}\rangle:A^\to \to A \times A$ is not.

We therefore introduce a whole new family of types, the \emph{displayed types}, which will allow us to consider our desired projection from the arrow category.
Indeed, while types in the category model correspond to opfibrations via Proposition \ref{prop:straight-icat}, displayed types will correspond to arbitrary functors over a base category.
First, we begin by recalling a well-known straightening-unstraightening result due to B\'enabou \cite{benabou2000distributors}.
\begin{proposition}[Straightening-unstraightening for displayed categories]\label{prop:dcat-model}
  Write $\Cat_\Prof$ for the bicategory that has as objects categories and as morphisms profunctors (not to be confused with the category $\Prof$ of Definition \ref{def:dprof}). Consider also the functor $[-,\Cat_\Prof]^{\textnormal{lax,normal}} : \Cat^\op \to \Cat$ that maps a category $A$ to the category of lax normal functors from $A$ to $\Cat_\Prof$.
  Write $\DCat := \Cat \ctxe^- [-,\Cat_\Prof]^{\textnormal{lax,normal}}$ for the category of \emph{displayed categories}.

  There exists a functor $(\ctxe^d)^*\, : \DCat \xrightarrow{\simeq} \Cat^\to$ which induces a fibred equivalence of categories over $\Cat$.
\end{proposition}
\begin{proof}[Proof sketch.]
  We follow Loregian \cite[Theorem 5.4.5]{loregianCoendCalculus2021}.
  Consider the forgetful functor $U: (\Cat_\Prof)_* \to \Cat_\Prof$, where $(\Cat_\Prof)_*$ is the bicategory that has as objects pointed categories and as morphisms pointed profunctors.
  The functor $(\ctxe^d)^*$ is defined on objects by mapping a lax normal functor $F : A \to \Cat_\Prof$ to the strict pullback $F^*U : A \times_{\Cat_\Prof} (\Cat_\Prof)_* \to A$.
\end{proof}

\begin{remark}
  The term ``displayed category'' is borrowed from Ahrens \& Lumsdaine work \cite{ahrensDisplayedCategories2019}, which studies functors to a base category in the context of type theory.
\end{remark}

\begin{definition}[Displayed category model]\label{def:dcat-model}
  The full split CC over $\Cat$ that has as its comprehension the unstraightening functor $(\ctxe^d)^*\, : \DCat \to \Cat^\to$ is called the \emph{displayed category model}.
\end{definition}

\begin{remark}
  The existence of the displayed category model is folklore. That is, it is known by researchers in the field, but, to our knowledge, there is no published work mentioning it.
\end{remark}

In order to use this second model, we now introduce a second copy of MLTT to our syntax.
For this copy, the typing judgement, term-in-type judgement, and context extension operation are written $\Gamma \vdash A\,\Ty_d$, $\Gamma \vdash a \dvar A$ and $\Gamma \cxe^d a: A$ respectively.
We call the types $\Gamma \vdash A\,\Ty_d$ the \emph{displayed types}.
Whenever we interpret these judgements, it will be using the displayed category model.

Now, we can relate the two copies of MLTT by the following proposition.

\begin{proposition}
  There is a strict family inclusion \cite{kovacs_generalized_2022} from the category model to the displayed category model.
  That is, the following two rules are validated.
  \begin{mathparpagebreakable}
    \inferrule*[right=$\Ty$\textnormal{\textsc{-Disp}}]{
      \Gamma \vdash A\,\Ty
    }{
      \Gamma \vdash A\,\Ty_d
    } \and
    \inferrule*[right=\textnormal{\textsc{Ctx-Ext-Disp}}]{
      \Gamma \vdash A\,\Ty
    }{
      \Gamma \cxe^d a : A \jeq \Gamma \cxe a : A
    }
  \end{mathparpagebreakable}
\end{proposition}
\begin{proof}
  The rule \textsc{$\Ty$-Disp} is validated by the embedding $\ICat_{\text{lax}} \to \DCat$ given by postcomposing a functor $A:\Gamma \to \Cat$ with the embedding $\Cat \hookrightarrow \Cat_\Prof$.
  The other rule follows by definition of $\cxe^d$.
\end{proof}

\begin{remark}
  An equivalent perspective, mentioned in op.\ cit., is to view strict family inclusions as degenerate modalities.
  Indeed, the rule \textsc{$\Ty$-Disp} can be seen as a dependent right adjoint \cite{birkedalModalDependentType2020} to the identity functor on the category of contexts.
\end{remark}

In the next section, we introduce the operation of ``twisting'' a type.
Crucially, the twist of the $\Hom$ type will allow us to recover our desired functor $\langle\mathsf{dom},\mathsf{cod}\rangle:A^\to \to A \times A$.

\section{Twisted types and dependent 2-sided fibrations}\label{sec:d2sfibs}
It was first observed by Street \cite{street_fibrations_1980} that for a category $A$, the functor $\langle \mathsf{dom},\mathsf{cod}\rangle: A^\to \to A\times A$ can be obtained from applying a general procedure to the profunctor $\hom_A:A^\op \times A \to \Set$.
Specifically, he introduced the Grothendieck construction on profunctors.
This operation takes an arbitrary profunctor $C:A^\op \times B \to \Set$ and gives back a new category $G(C)$ together with a span $A \xleftarrow{l_C} G(C) \xrightarrow{r_C} B$.
Further, he also introduced the notion of a 2-sided fibration \cite{streetFibrationsYonedasLemma1974}, and showed that his Grothendieck construction gives a fibred equivalence of categories from the category $\Prof$ of profunctors to the category $\mathsf{2SFib}$ of 2-sided fibrations.

Restricting our attention to types $A$ and $B$ in the empty context, we could internalize Street's Grothendieck construction to our syntax by stating that given a type $a : A^\op \cxe b : B \vdash C\,\Ty$, we can obtain a new displayed type $a :A \cxe b : B \vdash G(C)\,\Ty_d$.
We could then interpret $G(C)$ as the displayed category corresponding to the functor $G(C) \xrightarrow{\langle l_C,r_C \rangle} A \times B$.
However, in many of our cases of interest, the type $B$ depends on $A$; e.g.\ see Example \ref{ex:dhom}.
Hence, we need to consider a dependent version of Street's definitions and results.

We begin with a generalization of the notion of profunctors.

\begin{definition}[Dependent profunctors]\label{def:dprof}
  Write $\ICat_{\textnormal{str}} \defeq \Cat \ctxe^- [-,\Cat]$ for the category of indexed categories and strict natural transformations and let $K : (\ICat_{\textnormal{str}})^\op \to \Cat$ be the functor given by $K(A,B) \defeq [A \ctxe B^\op, \Cat]$.
  The \emph{category of dependent profunctors} $\DProf$ is the category $\ICat_{\textnormal{str}} \ctxe^- K$, whose objects are triples $(A,B,C)$, with $C : A \ctxe B^\op \to \Cat$.
\end{definition}

Indeed, dependent profunctors generalize profunctors in the following sense.

\begin{lemma}\label{lemma:dprof-generalizes}
  There is a fibred embedding $\Prof \hookrightarrow \DProf$.
\end{lemma}
\begin{proof}
  The embedding maps a profunctor $C:A \times B^\op \to \Cat$ to the dependent profunctor $C':A \ctxe (\mathsf{const}_B)^\op \xrightarrow{\cong} A \times B^\op \xrightarrow{ C} \Cat$.
\end{proof}

\begin{notation}
  Because of the previous lemma, we will conflate profunctors with their dependent version.
  Further, given two categories $A$ and $B$, we will write $A \ctxe B$ for $A \ctxe \mathsf{const}_B$.
\end{notation}

Note that dependent profunctors correspond to types $a : A \cxe b:B^\op \vdash C\,\Ty$ in the category model. More generally, they also correspond to types $\Gamma \cxe b : B^\op \vdash C\,\Ty$, by considering $\Gamma$ as one single category.
Indeed, this is why we chose the contravariant variable in the dependent profunctor to be the second one: with the opposite convention, a dependent profunctor would correspond to a type  $\Gamma^\op \cxe a : A \vdash B\,\Ty$.
While we will introduce the opposite of contexts later (Proposition \ref{prop:op-ctx}), they are not as well-behaved as opposite types.

\begin{example}\label{ex:dhom}
  For a type $\Gamma \vdash A\,\Ty$, the type $\Gamma \cxe b:A \cxe a:A^\op\vdash \Hom_A(a,b)\,\Ty$ induces the dependent profunctor $\Interp{\Hom_A(a,b)} : (\Interp{\Gamma} \ctxe \Interp{b:A}) \ctxe \Interp{a : A}^\op \to \Cat$.
\end{example}

Next, we generalize Street's notion of Grothendieck construction for profunctors.

\begin{definition}[Unstraightening of dependent profunctors]\label{def:unstraightening-dprof}
  Let $C:A \ctxe B^\op \to \Cat$ be a dependent profunctor.
  Define the functor $B \tye C : A \to \Cat$ by
  \[
    B \tye C \defeq \left\llbracket a:A \vdash \left(\sum_{b:B^\op}C^\op\right)^\op \,\Ty \right\rrbracket.
  \]
  The \emph{unstraightening} of $C$ is the category $A \ctxe (B \tye C)$.
\end{definition}

\begin{notation}
  By convention, $\tye$ binds more tightly than $\ctxe$.
\end{notation}

Concretely, the objects of $A \ctxe B \tye C$ are triples $(a,b,c)$ with $a\in A$, $b\in B(a)$ and $c\in C(a,b)$. A morphism $(a,b,c) \to (a',b',c')$ is a triple $(\alpha,\beta,\theta)$ with $\alpha:a \to a'$, $\beta: B(\alpha)(b) \to b'$ and $\theta: C(\alpha,\id_{B(\alpha)b})c \to C(\id_{a},\beta)c'$.
From this description, it's clear that there is a canonical projection $\pi_2:A \ctxe B \tye C \to A \ctxe B$ given by forgetting the last component.

\begin{proposition}[Twisted types]\label{prop:twists}
  The following rules are validated.
  \begin{mathparpagebreakable}
    \inferrule*[right=\textnormal{\textsc{Twist}}]{
      \Gamma \cxe b:B^\op \vdash C\,\Ty
    }{
      \Gamma \cxe b:B \vdash \mathsf{Tw}_b(C)\,\Ty_d
    } \and
    \inferrule*[right=\textnormal{\textsc{Twist-Weak}}]{
      \Gamma \vdash B\,\Ty \and \Gamma \vdash C\,\Ty
    }{
      \Gamma \cxe b:B \vdash \mathsf{Tw}_b(C) \jeq C\,\Ty_d
    }
  \end{mathparpagebreakable}
\end{proposition}
\begin{proof}
  The interpretation of $\mathsf{Tw}_b(C)$ is given by the displayed category corresponding to the canonical projection $\pi_2:A \ctxe B \tye C \to A \ctxe B$.
  The equality required by \textsc{Twist-Weak} readily follows.
\end{proof}

\begin{notation}
  We write $C[\bar b/b]$ for the type $\mathsf{Tw}_b(C)$.
  Notice how this notation makes \textsc{Twist-Weak} automatic: If $C$ does not depend on $b$ then $C[\bar b/b]$ is just $C$.
  The $C[\bar b/b]$ notation was inspired by Laretto et al.\ \cite{larettoDiDirectedFirstOrder2026}.
\end{notation}


In the following, it will be important to keep track of which displayed types come from directly applying the \textsc{Twist} rule. Hence, we introduce the following notation.

\begin{notation}\label{not:tscxe}
  For a type $\Gamma \cxe b : B^\op \vdash C \,\Ty$, we write $\Gamma \cxe b : B \tscxe c : C[\bar b/b]$ for the context $\Gamma \cxe b : B \cxe^d c : C[\bar b/b]$.
  In particular, uses of $\tscxe$ will imply that no variable in the context $\Gamma \cxe b : B$ has been (non-trivially) substituted. With this notation we have $\Interp{\Gamma \cxe b : B \tscxe c : C[\bar b/b]} \cong \Interp{\Gamma} \ctxe \Interp{B} \tye \Interp{C}$.
\end{notation}

We now return to our original motivation, and recover our desired projection $A^\to \to A\times A$.

\begin{example}\label{ex:sctxe-hom}
  Let $A$ be a category. By definition, the category $\Interp{b:A \cxe a:A \tscxe \hom_A(\bar{a},b)} \defeq A \ctxe A \tye \overline{\hom}_A$ has as objects triples $(b,a,f)$, with $f:a \to b$ in $A$, while a morphism $(b,a,f) \to (b',a',f')$ is a triple $(\beta, \alpha, e)$ with $\beta:b \to b'$, $\alpha: a \to a'$, and $e:\beta \circ f = f' \circ \alpha$. That is, this is the arrow category $A^\to$. The projection $\Interp{b:A \cxe a:A \tscxe \hom_A(\bar{a},b)} \to \Interp{b:A \cxe a:A}$ is then given by $\langle \mathsf{cod}, \mathsf{dom} \rangle:A^\to \to A \times A$.

  By contrast, the category $\Interp{b:A \cxe a:A^\op \cxe \hom_A(a,b)}$ has the same objects, but a morphism $(b,a,f)\to (b',a',f')$ is a triple $(\beta,\alpha,e)$ with $\beta:b \to b'$, $\alpha: a'\to a$ and $e:\beta \circ f \circ \alpha = f'$. That is, it is the twisted arrow category $A^{\mathsf{tw}}$. Hence, applying the \textsc{Twist} rule to the type $\Hom_A(a,b)$ appears to have made a ``twist''.
\end{example}



The generalization of a 2-sided fibration is given as follows.

\begin{definition}[D2SFib]\label{def:d2sfib}
  Consider a category $A$ and a functor $B:A \to \Cat$.
  A (split) \emph{dependent 2-sided fibration (D2SFib)} from $A$ to $B$ consists of the following data:
  \begin{enumerate}
    \item A category $C$;
    \item A functor $q:C \to A \ctxe B$;
    \item Data specifying that $q$ is a  ``local fibration''. That is, for each $a \in A$ the functor $q_{|a}:C_a \to B(a)$ given by postcomposing $q_a:C_a \to (A \ctxe B)_a$ with the isomorphisms $(A \ctxe B)_a \cong B(a)$ is a (split) fibration, so that each morphism $\beta: b \to qe$ in $B(a)$ has a specified cartesian lift $\beta^*e \to e$;
    \item Data specifying that $p \defeq \pi_A \circ q : C \to A$ is a (split) opfibration, so that each morphism $\alpha: pe \to a$ has a specified opcartesian lift $e \to \alpha_!e$;
  \end{enumerate}
  such that:
  \begin{enumerate}
      \setcounter{enumi}{4}
    \item $q$ is a cartesian functor of opfibrations from $p : C \to A$ to $\pi:A \ctxe B \to A$; and
    \item For each $\alpha: pe \to a$ in $A$ and $\beta: b \to qe$ in $B(p(e))$, the canonical morphism $\alpha_! \beta^* e \to (B(\alpha)\beta)^*\alpha_!e$ given by any of the universal properties is an identity.
  \end{enumerate}
\end{definition}

For simplicity, as they are our main objects of interest in this paper, we will focus only on split D2SFibs.
Therefore, we will leave the qualifier ``split'' implicit.
We will also identify a D2SFib as above with its first component, i.e.\ the functor $q:C \to A \ctxe B$.



\begin{definition}[Cartesian functors between D2SFibs]\label{def:d2sfib-1cells}
  A \emph{cartesian functor} between two D2SFibs $q:C \to A \ctxe B$ and $q' : C' \to A \ctxe B$ is a functor $\varphi:C \to C'$ satisfying the following two conditions.
  \begin{enumerate}
    \item It is a cartesian functor of opfibrations from $\pi_A \circ q: C \to A$ to $\pi_A \circ q': C' \to A$.
    \item For each $a \in A$, the restriction $\varphi_{a}:C_a \to C'_a$ is a cartesian functor of fibrations from $q_{|a}:C_a \to B(a)$ to $q'_{|a}:C'_a \to B(a)$
  \end{enumerate}
  Note that $q' \circ \varphi = q$ follows from these conditions.
\end{definition}

With this notion of cartesian functors, D2SFibs from $A$ to $B$ form a category $\DTSFib(A,B)$.
In fact, these assemble to a fibred category, as the following routine lemma shows.

\begin{lemma}
  Given a D2SFib $q:C \to A \ctxe B$ and a morphism $(F,G):(A',B')\to(A,B)$ in $\ICat_\textnormal{str}$, the pullback of $q$ along the cartesian functor ${\ctxe}^*(F,G):A' \ctxe B' \to A \ctxe B$ induces a D2SFib from $A'$ to $B'$.
  Further, this operation is pseudofunctorial on $\ICat_\textnormal{str}$, and hence it induces a non-split fibred category $\DTSFib$.
\end{lemma}

We are now ready to state our generalization of Street's equivalence between profunctors and 2-sided fibrations to the dependent case.

\begin{proposition}[Straightening-unstraightening for D2SFibs]\label{prop:straight-d2sfibs}
  Unstraightening of dependent profunctors (Definition \ref{def:unstraightening-dprof}) extends to a fibred equivalence of categories $\tye^* : \mathsf{DProf} \xrightarrow{\simeq} \DTSFib$ over $\ICat_{\textnormal{str}}$.
\end{proposition}
\begin{proof}[Proof sketch.]
  The straightening functor $S:\mathsf{D2SFib}(A,B) \to [A \ctxe B^\op,\Cat]$ is defined on objects by mapping a D2SFib $q:C \to A \ctxe B$ to the functor that maps an object $(a,b) \in A \ctxe B^\op$ to the fibre in $C$ above it.
  A complete proof is given in Section \ref{sec:staightd2sfibs} of the appendix.
\end{proof}

\begin{remark}
  It can be verified that a span $A \xleftarrow{p} C \xrightarrow{q} B$ is a 2-sided fibration, in Street's sense \cite{streetFibrationsYonedasLemma1974}, if and only if the induced functor $C \xrightarrow{\langle p,q\rangle} A \times B \xrightarrow{\cong} A \ctxe B$ is a D2SFib.
  Our work then generalizes Street's in the sense that we have the following commutative diagram.
  
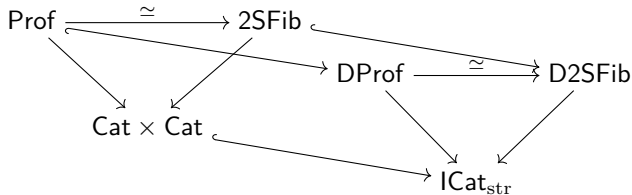
\[
    \begin{tikzcd}[cramped,sep=.7em]
      \Prof && \mathsf{2SFib} \\
      &&& \DProf && \DTSFib \\
      & {\Cat \times \Cat} \\
      &&&& \ICat_{\textnormal{str}}
      \arrow["\simeq", from=1-1, to=1-3]
      \arrow[hook', from=1-1, to=2-4]
      \arrow[from=1-1, to=3-2]
      \arrow[hook, from=1-3, to=2-6]
      \arrow[from=1-3, to=3-2]
      \arrow["\simeq", from=2-4, to=2-6]
      \arrow[from=2-4, to=4-5]
      \arrow[from=2-6, to=4-5]
      \arrow[hook, from=3-2, to=4-5]
    \end{tikzcd}
  \]
\end{remark}



By Proposition \ref{prop:straight-d2sfibs}, the canonical projection $\pi_2:\Interp{\Gamma \cxe b:B \tscxe c : C[\bar b/b]} \to \Interp{\Gamma \cxe b : B}$ is a D2SFib.
These D2SFibs will have an important role in TTT, as they are related the $\Hom$, $\Sigma$ and $\Pi$-types, as the next section will show.





\section{Twisted Type Theory}\label{sec:TTT}
In this section we introduce Twisted Type Theory.
We first introduce our new rules for $\Hom$-types. Afterwards, we introduce constructions already considered in the directed type theory literature, and relate them to our $\textsc{Twist}$ rule.
\subsection{\texorpdfstring{$\Hom$}{Hom}-types revisited}\label{sec:dispcat}

We begin by observing that we can now substitute $\Hom$-types along diagonals. That is, we can make the following derivation.
\begin{example}
  Let $A$ be a type in context $\Gamma$. Then $\Gamma \cxe a :A \vdash \Hom_A(\bar a, a)\,\Ty_d$ is derivable. Indeed, we have the following derivation.
  {
    \small
    \begin{mathparpagebreakable}
      \inferrule*[Right=$\Hom$-Form]{
        \Gamma \vdash A\,\Ty
      }{
        \inferrule*[Right=Twist]{
          \Gamma \cxe b:A \cxe a :A^\op \vdash \Hom_A(a, b)\,\Ty
        }{
          \inferrule*[Right=Subst]{
            \Gamma \cxe b:A \cxe a :A \vdash \Hom_A(\bar a, b)\,\Ty_d
          }{
            \Gamma \cxe a :A \vdash \Hom_A(\bar a, a)\,\Ty_d
          }
        }
      }
    \end{mathparpagebreakable}
  }
\end{example}

Our desired $\Hom$-introduction rule from desideratum (3) in the introduction is now given by the following proposition.
Note that it is analogous to the $\Id$-introduction rule of MLTT, which is given by a term $\Gamma, x:A \vdash \Refl_a:\Id(a,a)$.

\begin{proposition}[$\Hom$\textsc{-Intro}]\label{prop:hom-intro}
  The category model validates the following rule.
  {
    \small
    \begin{mathparpagebreakable}
      \inferrule*[Right=$\Hom$\textnormal{\textsc{-Intro}}]{
        \Gamma \vdash A\,\Ty
      }{
        \Gamma \cxe a :A \vdash \Refl_a \dvar \Hom_A(\bar a, a)
      }
    \end{mathparpagebreakable}
  }
\end{proposition}
\begin{proof}
  By an analogous argument to Example \ref{ex:sctxe-hom}, we see that $\Interp{\Gamma \cxe b:A \cxe a :A \tscxe \Hom_A(\bar a, b)}$ is isomorphic to $\Interp{\Gamma} \ctxe \Interp{A}^\to$, where $\Interp{A}^\to:\Gamma \to \Cat$ maps $\gamma \in \Gamma$ to $(A(\gamma))^\to$. The rule is then validated by the functor $\Interp{\Refl}:\Interp{\Gamma} \ctxe \Interp{A} \to \Interp{\Gamma} \ctxe \Interp{A}^\to$ given by $\Interp{\Refl}(\gamma,a) \defeq (\gamma, \id_a)$.
\end{proof}


Similarly, we could a give a $\Hom$-elimination rule that is analogous to the $\Id$-elimination rule.
However, the category model does not have all $\Pi$-types, and therefore this rule is not strong enough for our purposes.
Indeed, the lack of all $\Pi$-types had been previously observed \cite{licata_2-dimensional_2011,neumannSynthetic1CategoriesDirected2025}, and corresponds to the fact that the dependent product of an opfibration along an arbitrary functor may fail to be an opfibration.
Our approach to the $\Hom$-elimination rule therefore explicitly specifies that the motive can also depend on an additional variable of the appropriate form ($X$, in the rule below).
\begin{proposition}[$\Hom$\textsc{-Elim}]
  The category model validates the following rule.
  {
    \small
    \begin{mathparpagebreakable}
      \inferrule*[Right=$\Hom$\textnormal{\textsc{-Elim}}]
      {
        \Gamma \vdash A\,\Ty \\\\
        \Gamma \cxe a:A \vdash X\,\Ty \\\\
        \Gamma \cxe b: A \cxe a:A \tscxe f : \Hom_A(\bar a,b) \cxe x : X^\op\vdash D\,\Ty \\\\
        \Gamma \cxe a:A \cxe x:X \vdash d \dvar D[\bar{a}/b,\Refl_A/f, \bar{x}/x]
      }
      {
        \Gamma \cxe b: A \cxe a:A \tscxe f : \Hom_A(\bar a,b) \cxe x : X \vdash j(f,x,d) \dvar D
      }
    \end{mathparpagebreakable}
  }%
  It also validates the following rule, which has the same antecedents as \textnormal{\textsc{$\Hom$-Elim}}, and has the following conclusion.
  \[ \Gamma \cxe a:A \cxe x:X \vdash j(\Refl_a,x,d) \jeq d \dvar D[\bar{a}/b,\Refl_A/f, \bar{x}/x] \hspace{3em} \textnormal{\textsc{$\Hom$-Comp}}\]
\end{proposition}
\begin{proof}[Proof sketch.]
  For readability, we conflate here contexts, terms and types with their interpretation, e.g.\ $A=\Interp{A}$.
  For this proof, we introduce two functors $X[\mathsf{dom}]:\Gamma \ctxe A^\to \to \Cat$ and $r: \Gamma \ctxe A \ctxe X \to \Gamma \ctxe A^\to \ctxe X[\mathsf{dom}]$, given by $X[\mathsf{dom}](\gamma, f:a \to b)\defeq(\gamma, X(a))$ and $r(\gamma, a, x) \defeq (\gamma, \id_a, x)$, respectively.
  Recalling that $\Interp{\Gamma \cxe b:A \cxe a :A \tscxe \Hom_A(\bar a, b)}$ is isomorphic to $\Interp{\Gamma} \ctxe \Interp{A}^\to$ (proof of Proposition \ref{prop:hom-intro}), the proposition is equivalent to the existence of a functor $j$ making the rightmost diagram below commute.
  {
    \small
    \[
      \begin{tikzcd}[column sep=small]
        \bullet & {\Gamma \ctxe A^\to \ctxe X[\mathsf{dom}] \tye D } && {\Gamma \ctxe A \ctxe X} & {\Gamma \ctxe A^\to \ctxe X[\mathsf{dom}] \tye D } \\
        {\Gamma \ctxe A \ctxe X} & {\Gamma \ctxe A^\to \ctxe X[\mathsf{dom}]} && {\Gamma \ctxe A^\to \ctxe X[\mathsf{dom}]} & {\Gamma \ctxe A^\to \ctxe X[\mathsf{dom}]}
        \arrow[from=1-1, to=1-2]
        \arrow[from=1-1, to=2-1]
        \arrow["\lrcorner"{anchor=center, pos=0.085, rotate=0}, shift right=1.5, draw=none, from=1-1, to=2-2]
        \arrow["{\pi_2}"', from=1-2, to=2-2]
        \arrow["{d'}", from=1-4, to=1-5]
        \arrow["r"', from=1-4, to=2-4]
        \arrow["{\pi_2}", from=1-5, to=2-5]
        \arrow["d"{description}, shift left=3, curve={height=-6pt}, from=2-1, to=1-1]
        \arrow["{d'}"{description}, dashed, from=2-1, to=1-2]
        \arrow["r"', from=2-1, to=2-2]
        \arrow["j"{description}, dashed, from=2-4, to=1-5]
        \arrow["{=}"', no head, from=2-4, to=2-5]
      \end{tikzcd}
    \]
  }%
  Here, $d'$ is defined by the universal property of the pullback, as shown in the leftmost diagram.
  Further, $D$ is being considered as a dependent profunctor $D:\Gamma \ctxe A^\to \ctxe X[\mathsf{dom}]^\op \to \Cat$.

  We now define the action of $j$ on an object $(\gamma,f:a\to b,x) \in \Gamma \ctxe A^\to \ctxe X[\mathsf{dom}]$.
  First, writing $d^*=\pi_4 \circ d'$, note that $d^*(\gamma,a,x)$ is in $D(\gamma,\id_a,x)$.
  Further, observe that the there is a morphism from $\id_a$ to $f$ in $A^\to(\gamma)$, given by the tuple $(\id_a, f)$.
  It follows that $\phi \defeq (\id_\gamma, (\id_a, f), \id_x)$ is a morphism from $(\gamma, \id_a,x)$ to $(\gamma, f,x)$ in $\Gamma \ctxe A^\to \ctxe X[\mathsf{dom}]^\op$.
  Hence, we can define $j(\gamma,f:a\to b,x) \defeq (\gamma,f:a\to b,x, D(\phi)(d^*(\gamma,a,x)))$, which is readily seen to make the required diagram commute.
  The action of $j$ on morphisms and is similarly defined and is included in Section \ref{subsec:Hom-elim} of the appendix for completeness.
\end{proof}

Briefly, the elimination rule states that $r$ lifts against D2SFibs, in the sense made precise by the rightmost commutative diagram above.
This is analogous to how, for a path object $A \xrightarrow{r} A^I \xrightarrow{p} A \times A$ in a model category, the morphism $r$ lifts against fibrations \cite{awodeyHomotopyTheoreticModels2009}.

\subsection{Other operations in TTT}
In this section, we adapt rules from the directed type theory literature \cite{licata_2-dimensional_2011,neumannSynthetic1CategoriesDirected2025} to our current context, and then relate them to our new constructions.
We begin with the so-called deep polarization \cite{neumannSynthetic1CategoriesDirected2025}.
\begin{proposition}[\cite{licata_2-dimensional_2011,neumannSynthetic1CategoriesDirected2025}]\label{prop:op-ctx}
  The category model validates opposite contexts. That is, it validates the following rules.
  \begin{mathparpagebreakable}
    \inferrule*[Right=\textnormal{\textsc{Ctx-Op}}]
    {
      \Gamma\,\Ctx
    }
    {
      \Gamma^\op\,\Ctx
    }
    \and
    \inferrule*[Right=\textnormal{\textsc{Ctx-Op-Op}}]
    {
      \Gamma\,\Ctx
    }
    {
      (\Gamma^\op)^\op \jeq \Gamma \,\Ctx
    }
  \end{mathparpagebreakable}
  Semantically, $\Interp{\Gamma^\op}$ is interpreted as $\Interp{\Gamma}^\op$.
\end{proposition}

We can now capture syntactically the contravariant Grothendieck construction as follows.

\begin{notation}[\cite{neumannSynthetic1CategoriesDirected2025}]
  For a type $\Gamma^\op \vdash A\,\Ty$, we write $\Gamma \cxe^- A$ for the context $(\Gamma^\op \cxe A^\op)^\op$. It follows that $\Interp{\Gamma \cxe^- A} = \Interp{\Gamma} \ctxe^- \Interp{A}$.
\end{notation}

We also recall the existence of a universe of discrete types.
The same points of Remark \ref{remark:univs} apply to these universes.
\begin{proposition}[\cite{licata_2-dimensional_2011,neumannSynthetic1CategoriesDirected2025}]
  Under the assumption of a countable hierarchy of inaccessible cardinals, the category model validates countable universes $\Set_i$ classifying discrete categories.
  For these \emph{discrete types} $\Gamma \vdash A:\Set_i$, the category model validates the judgemental equality $A^\op \jeq A\,\Ty$, as well as discreteness of $\Hom$-types, i.e.\ $\Hom_A(a,b) : \Set_i$.
\end{proposition}

Next, we introduce $\Pi$-types. For conciseness, here and in the next propositions we only state the bijection between terms in the category models.
The rules that follow from these bijections can be found in Section \ref{sec:add-rules} of the appendix.
\begin{notation}
  Write $\Tm(\Gamma, A)$ (resp. $\Tm_d(\Gamma, A)$) for the set of terms of type $A$ in context $\Gamma$ in the category model (resp. displayed category model). Recall that this set is defined as the set of sections of the canonical projection $\pi:\Interp{\Gamma \cxe a : A} \to \Interp{\Gamma}$ (resp. $\pi:\Interp{\Gamma \cxe^d a : A} \to \Interp{\Gamma}$).
\end{notation}

\begin{proposition}[\cite{licata_2-dimensional_2011,neumannSynthetic1CategoriesDirected2025}]\label{prop:pi-types}
  Given types $\Gamma^\op \vdash A\,\Ty$ and $\Gamma \cxe^- a : A \vdash B\,\Ty$, the category model has the $\Pi$-type $\Gamma \vdash \Pi_AB\,\Ty$, where its terms satisfy the following natural isomorphism.
  \[ \mathsf{app} : \Tm(\Gamma, \Pi_AB) \cong \Tm(\Gamma \cxe^- a : A, B) : \mathsf{\lambda}\]
\end{proposition}

As usual, if the type $B$ does not depend on $A$, we write $A\to B$ instead of $\prod_AB$.

Now, consider the type $\Gamma \cxe B : \UU \cxe A^\op : \UU \vdash A \to B \,\Ty$ and apply the \textsc{Twist} rule to obtain the type $\Gamma \cxe B : \UU \cxe A : \UU \vdash \bar{A} \to B \,\Ty_d$.
The following proposition characterizes the terms of this type in the category model, after substituting for some types $A$ and $B$.

\begin{proposition}\label{prop:function-twist}
  Let $\Gamma \vdash A\,\Ty$ and $\Gamma \vdash B\,\Ty$ be two types. Then we have the following natural isomorphism of terms.
  \[ \mathsf{app} : \Tm_d(\Gamma, \bar{A} \to B) \cong \Tm(\Gamma \cxe a : A, B) : \mathsf{\lambda}\]
\end{proposition}
\begin{proof}
  Analogous to Proposition \ref{prop:pi-types}.
\end{proof}

Twists also allow characterizing the opposite of a $\Sigma$-type, as the next proposition shows.

\begin{proposition}\label{prop:op-sigma}
  Let $\Gamma \vdash A\,\Ty$ and $\Gamma \cxe a:A \vdash B\,\Ty$ be two types. Then we have the following natural isomorphism.
  \[ \langle \mathsf{fst},\mathsf{snd}\rangle : \Tm\left(\Gamma, z : \bigg(\sum_{a:A}B\bigg)^\op\right) \cong
  \sum_{x :\Tm(\Gamma, A^\op)} \Tm_d(\Gamma, B^\op[\bar x/a]) : \mathsf{pair} \]
  Here, the type $\Gamma \vdash B^\op[\bar x/a]$ has been obtained by substituting $x$ for $a$ in the twisted type $\Gamma \cxe a : A^\op \vdash B[\bar a/a]\,\Ty_d$.
\end{proposition}
\begin{proof}
  The bijection follows from the isomorphism $\Interp{\Gamma \cxe A^\op \tscxe B^\op} \cong \Interp{\Gamma} \cxe \Interp{(\sum_{a:A}B)^\op}$.
\end{proof}

A complete description of TTT is given in Appendix \ref{sec:add-rules}.
It includes additional rules omitted here for reasons of space, e.g.\ the empty context $\diamond$ satisfies $\diamond^\op \jeq \diamond$.
Nonetheless, a concise definition is given as follows.

\begin{definition}
  \emph{Twisted Type Theory (TTT)} is MLTT extended with the $\Sigma$-types, $\Pi$-types, $\Hom$-types, twists, opposite types, displayed types, universes and opposite contexts, as described this paper. The interpretation of TTT given is called the \emph{category model} of TTT.
\end{definition}

\section{Applications}\label{sec:applications}
In this section, we use TTT to develop some synthetic category theory.
Any mentions to an interpretation of the syntax will therefore be about the category model of TTT.
\subsection{Directed univalence}
By definition, the semantics of the type of functions $A \to B$ is the category of functors $[\llbracket A \rrbracket, \llbracket B \rrbracket]$.
When $A$ and $B$ are sets, this category has no non-identity morphisms; i.e.\ it is a set.
This suggests the following definition.

\begin{definition}\label{def:dunivalence}
  For sets $\Gamma^\op \vdash A : \Set$ and $\Gamma \vdash B : \Set$, the conversion rule $\Gamma \vdash A \to B \equiv \Hom_\Set(A,B) : \Set$ is called \emph{directed univalence}.
\end{definition}

By the previous discussion, we obtain the following proposition.

\begin{proposition}\label{prop:dunivalence}
  The category model validates directed univalence.
\end{proposition}

An alternative to the judgemental equality $\Hom_\Set(A,B) : \Set$ would be requiring that the canonical function $\Hom_\Set(A,B) \to (A \to B)$ induces an equivalence of types.
However, such a function can only be formed when $A$ and $B$ are types in the empty context.
Our current definition applies in more general contexts, as well as being more useful in pen and paper proofs, at the cost of severely complicating typechecking.

We will freely use directed univalence in the next sections.


\subsection{Natural transformations}\label{subsec:nats}

As was stated in the introduction, twisted types allow us to capture syntactically the semantic notion of natural transformations between functors.

\begin{proposition}\label{prop:internal-nat-hom}
  For two terms $\IsTm[a:A]{Fa}{B}$ and $\IsTm[a:A]{Ga}{B}$, terms $\IsTmD[a:A]{t}{\Hom(\overline{Fa},Ga)}$ correspond to natural transformations $\Interp{F} \to \Interp{G}$ in the category model.
\end{proposition}
\begin{proof}
  By the universal property of the pullback, terms $t$ as above correspond to functors $\tau:\Interp{A} \to \Interp{A}^\to$ satisfying that for all $a \in \Interp{A}$, the morphism $\tau(a)$ has as domain $\Interp{F}(a)$ and as codomain $\Interp{G}(a)$.
  That is, they correspond to natural transformations $\Interp{F} \to \Interp{G}$.
\end{proof}

\begin{remark}\label{remark:lax-ends}
  For the reader familiar with lax ends, the previous proposition follows from a more general fact, which is readily verified: for a profunctor $F:A \times A^\op \to \Cat$, the category of lifts from $A$ to $A \ctxe A \tye F$ through the diagonal $\delta : A \to A \times A$ is the lax end of $F$. The previous result then follows by setting $F$ to be $\Hom(F-,G-):A \times A^\op \to \Cat$.
\end{remark}

Generalizing to arbitrary contexts and dependent functions, an argument analogous to the previous one gives the following result.
\begin{proposition}\label{prop:dfunext}
  Let $\Gamma \vdash F,G:\prod_AB$ be two functions.
  Then the canonical function $\Tm_d(\Gamma, \hom_{\Pi_AB}(\bar F,G)) \to \Tm_d(\Gamma \cxe^- a :A, \hom_{Ba}(\overline{F(a)}, G(a))$
    given by $\Hom$-elimination induces a bijection between the interpretation of both sets.
  \end{proposition}

  Internalizing this result, we present the following axiom.
  \begin{definition}
    Let $\Gamma \vdash F,G:\prod_AB$ be two functions.
    The axiom \emph{directed function extensionality} states that given a term $\Gamma \cxe^- a :A \vdash t : \hom_{Ba}(\overline{F(a)}, G(a))$, we can form a new term $\Gamma \vdash \mathsf{dfunext}(t) : \hom_{\Pi_AB}(F,G)$.
  \end{definition}

  Note that the natural transformations are syntactically analogous to the homotopies $\IsTm[a:A]{t}{\Id(Fa,Ga)}$ in HoTT.
  However, unlike in that setting, we cannot form the type of all natural transformations by using a $\Pi$-type, since we only have $\Hom(\overline{Fa},Ga)\,\Ty_d$. The problem is resolved for functors into $\Set$, by uncurrying, as the next proposition shows.

  \begin{proposition}\label{prop:internal-nat-pi}
    Let $A$ be a type and let $F$ and $G$ be copresheaves from $A$, i.e.\ $a:A \vdash Fa : \Set$ and $a:A \vdash Ga : \Set$. The interpretation of the type
    \[ \Nat(F,G) \defeq \prod_{(a,x):\sum_{a:A}Fa}Ga \]
    is the set of natural transformations $F \to G$.
  \end{proposition}
  \begin{proof}
    Indeed, we have the following chain of isomorphisms.
    \begin{align*}
      \Tm\bigg(\diamond, \prod_{(a,x):\sum_{a:A}Fa}Ga\bigg) &\cong
      \Tm\big(\sum_{a:A}Fa, Ga\big)  && \text{by Proposition \ref{prop:pi-types}}\\[-.5em]
      & \cong
      \Tm(a:A \cxe x: Fa, Ga) && \text{by Proposition \ref{prop:Sigma}}\\
      & \cong
      \Tm_d(a:A, \overline{Fa} \to Ga) && \text{by Proposition \ref{prop:function-twist}} \\
      &\cong
      \Tm_d(a:A, \Hom(\overline{Fa}, Ga))  && \text{by Proposition \ref{prop:dunivalence}} \\
      &\cong \Nat(\Interp{F},\Interp{G}) && \text{by Proposition \ref{prop:internal-nat-hom}} \qedhere
    \end{align*}
  \end{proof}

  When working with copresheaves, we can also consider equalities of natural transformations. In the following, for ease of reading, for a natural transformation $\varphi : \Nat(F,G)$ we write $\varphi_a(x)$ for $\varphi(a,x)$.

  \begin{proposition}
    Let $A$ be a type, $F$ and $G$ be two copresheaves from $A$, and $\varphi,\psi:\Nat(F,G)$ be two natural transformations. Then, the judgement
    \[ a:A \cxe x:Fa \vdash t:\Hom_{Ga}(\varphi_a(x),\psi_a(x))\]
    typechecks and its interpretation is an equality of natural transformations $\varphi$ and $\psi$.
  \end{proposition}
  \begin{proof}
    Since $Ga$ is a set, we have $(Ga)^\op\jeq Ga$. Therefore, we have $a:A \cxe x : Fa \vdash \varphi_a(x) : (Ga)^\op$ and so the judgement in the proposition typechecks. Further, the $\Hom$-type of a set is trivial, hence we have that $\varphi_a(x)$ and $\psi_a(x)$ are equal for each $x:Fa$, i.e.\ $\varphi=\psi$.
  \end{proof}

  \begin{corollary}\label{cor:nat-isos}
    Let $A$ be a type and $F$ and $G$ be two copresheaves from $A$. Then the interpretation of a tuple $(\varphi,\psi,t,r)$ satisfying
    \begin{gather*}
      \varphi:\Nat(F,G) \  \psi:\Nat(G,F) \\
      a:A \cxe x : Fa \vdash t : \Hom_{Fa}(x, \psi_a(\varphi_a(x))) \\
      a:A \cxe y : Ga \vdash t : \Hom_{Ga}(\varphi_a(\psi_a(y)), y)
    \end{gather*}
    is an isomorphism of $F$ and $G$.
  \end{corollary}

  \subsection{Yoneda lemma}
  In this section, we sketch a proof of the Yoneda lemma. A complete proof with the omitted details can be found in Section \ref{subsec:yoneda-full} in the appendix.

  % ...

  \begin{lemma}[Yoneda]\label{lemma:yoneda}
    Let $A:\UU$ be a type in the empty context. Then the two following functors are naturally isomorphic.
    \begin{align*}
      Y, \bar Y   & :(A \to \Set) \times A \to \Set             \\
      Y(F,a)      & \defeq \Nat(\Hom_A(a,-), F) \\
      \bar Y(F,a) & \defeq Fa
    \end{align*}
  \end{lemma}
  \begin{proof}[Proof sketch]
    We use Corollary \ref{cor:nat-isos}. First, we define a natural transformation $\Phi:Y \to \bar Y$ by evaluation on $\Refl_a$, as follows.
    {
      \small
      \begin{mathpar}
        \inferrule*[Right=$\Pi$-Elim]
        {
          F:A \to \Set \cxe a :A \cxe \varphi: Y(F,a) \vdash (a,\Refl_a) \cvar \sum_{x:A}\hom(a,x)
        }
        {
          F:A \to \Set \cxe a :A \cxe \varphi: Y(F,a)\vdash \varphi_a(\Refl_a) : Fa
        }
      \end{mathpar}
    }
    The first step in the derivation of $\Psi:\bar Y \to Y$ is given by $\Hom$-elimation, as follows.
    {
      \small
      \begin{mathpar}
        \inferrule*[Right=$\Hom$-Elim]
        {
          F:A \to \Set \cxe a:A \vdash Fa\,\Ty \\\\
          F:A \to \Set \cxe x:A \cxe a:A \tscxe f : \hom(a,x), r:Fa \vdash Fx\,\Ty \\\\
          F:A \to \Set \cxe a :A \cxe r:Fa \vdash r : Fa
        }
        {
          F:A \to \Set \cxe x:A \cxe a:A \tscxe f : \hom(a,x) \cxe r:Fa \vdash j(f,r,r) : Fx
        }
      \end{mathpar}
    }%
    By rearranging the context, we obtain the following term.
    {
      \small
      \[
        F:A \to \Set \cxe a :A \cxe r:Fa \vdash \lambda (x,f) . j(f,r,r) : \prod_{(x,f):\sum_{x:A}\hom(a,x)}Fx
      \]
    }%
    That $\Phi \circ \Psi$ is the identity follows from the computation rule for $\Hom$-types, while $\Psi \circ \Phi = \id$ follows from another application of $\Hom$-elimination.
  \end{proof}

  \section{Conclusion and future work}\label{sec:future-work}
  We have introduced TTT as an extension of MLTT and showed that it allows for synthetic reasoning about categories.
  While we argued for the three desiderata given in the introduction, we emphasize that any such criteria will lead to different advantages and disadvantages.
  For example, a negative consequence of the desideratum (1), i.e.\ that closed types are interpreted as categories, is that not all $\Pi$-types will exist, and so the $\Pi$-formation rule has to be restricted, as we did in Proposition \ref{prop:pi-types}.
  Hence, we see our work as exploring and developing one of many different possibilities in the space of directed type theories.

  Nonetheless, we argue that the introduction of twisted types is an important development:
  They are linked to $\Sigma$-types and $\Pi$-types via Propositions \ref{prop:op-sigma} and \ref{prop:function-twist}, as well as to lax ends via Remark \ref{remark:lax-ends}.
  Further, the $\Hom$-elimination rule it allows is essentially a generalization of the Yoneda lemma.
  Indeed, careful inspection of the proof given in Lemma \ref{lemma:yoneda} reveals that the first step in the derivation of the natural transformation $\Psi: \bar Y \to Y$ already carries all the information of the Yoneda lemma, albeit in a disorganized way: the next steps in the proof are all invertible, and hence are just reorganizing the data.

  This being said, there are multiple avenues to explore which could improve TTT.
  For one, we need to develop a richer theory of D2SFibs.
  We expect to find a mathematical structure generalizing weak factorization systems, which would allow us to generalize North's \cite{north_type_2019} results about 2-sided fibrations to the dependent case.
  We expect that such results would clarify and refine our $\Hom$-elimination principle.
  An additional area to explore is the reintroduction of the $\mathsf{core}$-types from North \cite{north_towards_2019}.
  These types are useful as not every categorical construction is either covariant or contravariant. Indeed, in our statement of Yoneda we required that $A$ be a closed type precisely for this reason.
  Furthermore, the rules for permuting contexts involving twisted types could be improved.
  Our proof of the Yoneda lemma uses such a permutation, but our current proof of its correctness is quite technical (see Lemma \ref{lemma:yoneda-perm-proof}).
  Finally, future work should explore the higher dimensional models of TTT, and study the properties the twisted types have in this context.



  \bibliography{references}

  \appendix

  \section{Complete description of TTT}\label{sec:add-rules}
  Twisted Type Theory has two copies of MLTT, see e.g.\ Hofmann \cite{hofmannSyntaxSemanticsDependent1997} for survey of MLTT.
  Of note is that we assume the structural rules, i.e.\ weakening, substitution, and permutation; as they would not be admissible given our following rules.
  We write $\Gamma \vdash A\,\Ty$, $\Gamma \vdash a : A$ and $\Gamma \cxe a: A$ (resp. $\Gamma \vdash A\,\Ty_d$, $\Gamma \vdash a \dvar A$ and $\Gamma \cxe^d a: A$) for the typing judgement, term-in-type judgement, and context extension operation respectively of the first MLTT copy (resp. second MLTT copy).
  These two copies of MLTT are related via the following rules.

  \begin{figure}[h]
    \begin{rules}
      \begin{mathparpagebreakable}
        \inferrule*[right=$\Ty$\textnormal{\textsc{-Disp}}]{
          \Gamma \vdash A\,\Ty
        }{
          \Gamma \vdash A\,\Ty_d
        } \and
        \inferrule*[right=\textnormal{\textsc{Ctx-Ext-Disp}}]{
          \Gamma \vdash A\,\Ty
        }{
          \Gamma \cxe^d a : A \jeq \Gamma \cxe a : A
        } \and
        \inferrule*[right=\textnormal{\textsc{Twist}}]{
          \Gamma \cxe b:B^\op \vdash C\,\Ty
        }{
          \Gamma \cxe b:B \vdash \mathsf{Tw}_b(C) \,\Ty_d
        } \and
        \inferrule*[right=\textnormal{\textsc{Twist-Weak}}]{
          \Gamma \vdash B\,\Ty \and \Gamma \vdash C\,\Ty
        }{
          \Gamma \cxe b:B \vdash \mathsf{Tw}_b(C) \jeq C\,\Ty_d
        }
      \end{mathparpagebreakable}
    \end{rules}
    \caption{Rules about displayed types and twists.}
  \end{figure}

  We assume two countable hierarchies of Coquand universes \cite{coquandPresheafModelType2013,gratzer_multimodal_2021}.
  We write $\UU_i$ for the universes of types and $\Set_i$ for the universes of discrete types.
  \cite{kovacs_generalized_2022}.
  These universes are related by a strict family inclusion from the discrete types in $\Set_i$ to the base types in $\UU_i$, which commutes with the strict family inclusion in each hierarchy.
  As usual, we use these universes as if they were Russell universes, and leave the indices implicit.
  The discrete types $\Gamma \vdash A:\Set$ are postulated to satisfy $\Gamma \vdash A^\op \jeq A : \Set$.

  We additionally have rules concerning opposite types and opposite contexts. These had already appeared in other directed type theories in the literature \cite{licata_2-dimensional_2011,neumannSynthetic1CategoriesDirected2025}.


  \begin{figure}[h]
    \begin{rules}
      \begin{mathparpagebreakable}
        \inferrule*[]{
          \IsTy{A}
        }{
          \IsTy{A^\op}
        } \and
        \inferrule*[]{
          \IsTy{A}
        }{
          \EqTy{(A^\op)^\op}{A}
        } \\
        \inferrule*[]
        {
          \Gamma\,\Ctx
        }
        {
          \Gamma^\op\,\Ctx
        }
        \and
        \inferrule*[]
        {
          \Gamma\,\Ctx
        }
        {
          (\Gamma^\op)^\op \jeq \Gamma \,\Ctx
        } \and
        \inferrule*[]{
        }{
          \diamond^\op \jeq \diamond
        } \and
        \inferrule*[]{
          \diamond \vdash A\ \Ty \and \Gamma\,\Ctx
        }{
          \Gamma \cxe A \jeq \Gamma \cxe^- A
        }
      \end{mathparpagebreakable}
    \end{rules}
    \caption{Rules about opposite types and contexts.}
  \end{figure}

  \begin{notation}
    Let $\Gamma^\op \vdash B\,\Ty$ and $\Gamma \cxe b:B^\op \vdash C\,\Ty$ be types.
    \begin{enumerate}
      \item Other than in the $\textsc{Twist}$ and $\textsc{Twist-Weak}$ rules, we write $C[\bar b/b]$ for $\mathsf{Tw}_b(C)$.
      \item We write $\Gamma \cxe^- B$ for the context $(\Gamma^\op \cxe B^\op)^\op$.
      \item We write $\Gamma \cxe b:B \tscxe c : C[\bar b/b]$ for the context $\Gamma \cxe b:B \cxe^d c : C[\bar b/b]$.
        This notation has as presupposition that $\Gamma \cxe b:B^\op \vdash C\,\Ty$ and so, e.g., won't be used after a (non-trivial) substitution.
      \item We write $\sum_{b:B}^\tscxe c : C[\bar b/b]$ for the type $(\sum_{b:B^\op}C^\op)^\op$. We only use this in the appendix.
      \item We write $X \to Y$ and $X \times Y$ instead of $\Pi_{x:X}Y$ and $\Sigma_{x:X}Y$, respectively, when $Y$ does not depend on $X$.
    \end{enumerate}
  \end{notation}

  Finally, TTT includes rules for $\Sigma$, $\Pi$, and $\Hom$-types.
  These are given in the Figures \ref{fig:sigma}, \ref{fig:pi}, and \ref{fig:hom} respectively.
  This concludes the presentation of TTT.

  \begin{remark}
    The rules about $\Sigma$-types and $\Pi$-types not involving twists or opposites have been adapted from previous work in the literature \cite{licata_2-dimensional_2011,neumannSynthetic1CategoriesDirected2025}.
  \end{remark}

  \begin{remark}
    The last two rules for $\Sigma$-types in Figure \ref{fig:sigma}, involving the $\mathsf{swap}$ and $\mathsf{swap}^{-1}$ operations, are only used in the proof Lemma \ref{lemma:yoneda-perm-proof}.
    They are a technical tool to facilitate rearrangements of the data in contexts with twisted types.
  \end{remark}


  \begin{figure}[h!]
    \begin{rules}
      \begin{mathparpagebreakable}
        \inferrule*[]{
          \Gamma \vdash A\,\Ty \and \Gamma \cxe a : A \vdash B\,\Ty
        }{
          \Gamma \vdash \Sigma_{a:A}B \,\Ty
        }  \and
        \inferrule*[]{
          \Gamma \vdash A\,\Ty \and \Gamma \cxe a : A \vdash B\,\Ty \and \Gamma \vdash p : \Sigma_{a:A}B
        }{
          \Gamma \vdash \pi_1(p) : A \and \Gamma \vdash \pi_2(p) : B[\pi_1(p)/a] \and \Gamma \vdash (\pi_1(p),\pi_2(p)) \jeq p : \Sigma_{a:A}B
        } \and
        \inferrule*[]{
          \Gamma \vdash A\,\Ty \and \Gamma \cxe a : A \vdash B\,\Ty \and \Gamma \vdash u : A \and \Gamma \vdash v : B
        }{
          \Gamma \vdash (u,v) : \Sigma_{a:A}B \and
          \Gamma \vdash \pi_1(u,v) \jeq u : A \and \Gamma \vdash \pi_2(u,v) \jeq v : B[u/a]
        } \and
        \inferrule*[]{
          \Gamma \vdash A\,\Ty \and \Gamma \cxe a : A \vdash B\,\Ty \and \Gamma \vdash p : (\Sigma_{a:A}B)^\op
        }{
          \Gamma \vdash \pi_1(p) : A^\op \and \Gamma \vdash \pi_2(p) : B^\op[\overline{\pi_1(p)}/a] \and \Gamma \vdash (\pi_1(p),\pi_2(p)) \jeq p : \Sigma_{a:A}B
        } \and
        \inferrule*[]{
          \Gamma \vdash A\,\Ty \and \Gamma \cxe a : A \vdash B\,\Ty \and \Gamma \vdash u : A^\op \and \Gamma \vdash v : B^\op[\bar u/a]
        }{
          \Gamma \vdash (u,v) : (\Sigma_{a:A}B)^\op \and
          \Gamma \vdash \pi_1(u,v) \jeq u : A^\op \and \Gamma \vdash \pi_2(u,v) \jeq v : B^\op[\bar u/a]
        }\and
        \inferrule*[]{
          \Gamma \vdash X,Y\,\Ty \and \Gamma \cxe x : X \cxe y : Y^\op \vdash Z\,\Ty \and \Gamma \vdash p : (\Sigma_{x:X}\Sigma_{y:Y}^\tscxe Z(x,\bar y))^\op
        }{
          \Gamma \vdash \mathsf{swap}(p) : \Sigma_{y:Y^\op}\Sigma_{x:X^\op}^\tscxe Z(\bar x,y)^\op
          \and
          \Gamma \vdash \mathsf{swap}^{-1}(\mathsf{swap}(p) ): (\Sigma_{x:X}\Sigma_{y:Y}^\tscxe Z(x,\bar y))^\op
        } \and
        \inferrule*[]{
          \Gamma \vdash X,Y\,\Ty \and \Gamma \cxe x : X \cxe y : Y^\op \vdash Z\,\Ty \and \Gamma \vdash q : \Sigma_{y:Y^\op}\Sigma_{x:X^\op}^\tscxe Z(\bar x,y)^\op
        }{
          \Gamma \vdash \mathsf{swap}^{-1}(q) : (\Sigma_{x:X}\Sigma_{y:Y}^\tscxe Z(x,\bar y))^\op
          \and
          \Gamma \vdash \mathsf{swap}(\mathsf{swap}^{-1}(q)): \Sigma_{y:Y^\op}\Sigma_{x:X^\op}^\tscxe Z(\bar x,y)^\op
        }
      \end{mathparpagebreakable}
    \end{rules}
    \caption{Rules about $\Sigma$-types.}
    \label{fig:sigma}
  \end{figure}

  \begin{figure}[h!]
    \begin{rules}
      \begin{mathparpagebreakable}
        \inferrule*[]{
          \Gamma^\op \vdash A\,\Ty \and \Gamma \cxe^- a : A \vdash B\,\Ty
        }{
          \Gamma \vdash \Pi_{a:A}B \,\Ty
        }  \and
        \inferrule*[]{
          \Gamma^\op \vdash A\,\Ty \and \Gamma \cxe^- a : A \vdash B\,\Ty \and \Gamma \vdash f : \Pi_{a:A}B
        }{
          \Gamma \cxe^- a : A\vdash f(a) : B \and
          \Gamma \vdash \lambda a.f(a) \jeq f : \Pi_{a:A}B
        } \and
        \inferrule*[]{
          \Gamma^\op \vdash A\,\Ty \and \Gamma \cxe^- a : A \vdash B\,\Ty \and \Gamma \cxe^- a : A \vdash b : B
        }{
          \Gamma \vdash \lambda a. b : \Pi_{a:A}B \and
          \Gamma \cxe^- a : A \vdash (\lambda a.b)(a) \jeq b : B
        } \and
        \inferrule*[]{
          \Gamma \vdash A\,\Ty \and \Gamma \vdash B\,\Ty \and \Gamma \vdash f \dvar \bar A \to B
        }{
          \Gamma \cxe a : A  \vdash f(a) : B \and \Gamma \vdash \lambda a.f(a) \jeq f : \bar A \to B
        } \and
        \inferrule*[]{
          \Gamma \vdash A\,\Ty \and \Gamma \vdash B\,\Ty \and \Gamma \cxe a :A \vdash b : B
        }{
          \Gamma \vdash \lambda a. b : \bar A \to B \and
          \Gamma \cxe a : A \vdash (\lambda a.b)(a) \jeq b : B
        }
      \end{mathparpagebreakable}
    \end{rules}
    \caption{Rules about $\Pi$-types.}
    \label{fig:pi}
  \end{figure}

  \begin{figure}[h!]
    \begin{rules}
      \begin{mathparpagebreakable}
        \inferrule*[right=$\Hom$\textnormal{\textsc{-Form}}]{
          \IsTy{A}
        }{
          \Gamma \cxe b:A \cxe a:A^\op \vdash \Hom_{A}(a,b) : \Set
        }  \and
        \inferrule*[right=$\Hom$\textnormal{\textsc{-Intro}}]{
          \Gamma \vdash A\,\Ty
        }{
          \Gamma, a :A \vdash \Refl_a \dvar \Hom_A(\bar a, a)
        } \and
        \inferrule*[right=$\Hom$\textnormal{\textsc{-Elim}}]
        {
          \Gamma \vdash A\,\Ty \\\\
          \Gamma \cxe a:A \vdash X\,\Ty \\\\
          \Gamma \cxe b: A \cxe a:A \tscxe f : \Hom_A(\bar a,b) \cxe x : X^\op\vdash D\,\Ty \\\\
          \Gamma \cxe a:A \cxe x:X \vdash d \dvar D[\bar{a}/b,\Refl_A/f, \bar{x}/x]
        }
        {
          \Gamma \cxe b: A \cxe a:A \tscxe f : \Hom_A(\bar a,b) \cxe x : X \vdash j(f,x,d) \dvar D
        } \and
        \inferrule*[right=$\Hom$\textnormal{\textsc{-Comp}}]
        {
          \Gamma \vdash A\,\Ty \\\\
          \Gamma \cxe a:A \vdash X\,\Ty \\\\
          \Gamma \cxe b: A \cxe a:A \tscxe f : \Hom_A(\bar a,b) \cxe x : X^\op\vdash D\,\Ty \\\\
          \Gamma \cxe a:A \cxe x:X \vdash d \dvar D[\bar{a}/b,\Refl_A/f, \bar{x}/x]
        }
        {
          \Gamma \cxe a:A \cxe x:X \vdash j(\Refl_a,x,d) \jeq d \dvar D[\bar{a}/b,\Refl_A/f, \bar{x}/x]
        }
      \end{mathparpagebreakable}
    \end{rules}
    \caption{Rules about $\Hom$-types.}
    \label{fig:hom}
  \end{figure}


  \section{Straightening and unstraightening for D2SFibs}\label{sec:staightd2sfibs}
  In this section we give a complete proof of the equivalence between dependent profunctors and D2SFibs, stated in Proposition \ref{prop:straight-d2sfibs}.
  For the rest of this section, let $B: A \to \Cat$ be an indexed category.

  \subsection{The unstraightening functor}
  We first show that unstraightening of a dependent profunctor is indeed a D2SFib.
  \begin{proposition}
    Let $C:A \ctxe B^\op \to \Cat$ be a dependent profunctor, then the canonical projection $\pi_2:A \ctxe B \tye C \to A \ctxe B$ is a D2SFib.
  \end{proposition}
  \begin{proof}
    We show that conditions (3) to (6) of Definition \ref{def:d2sfib} are satisfied.
    \begin{enumerate}
        \setcounter{enumi}{2}
      \item We show that for each $a \in A$, the functor $(\pi_2)_{|a}: (A \ctxe B \tye C)_a \to B(a)$ is a fibration.
        Let $\beta: b \to (\pi_2)_{|a}(a,b',c)$ be a morphism in $B(a)$. We define the opcartesian lift of $\beta$ to be
        \[ \beta^* \defeq (\id_a, \beta, \id_{C(\id_a,\beta)c}): (a,b,C(\id_a,\beta)c) \to (a,b',c). \]
      \item We show that $\pi:A \ctxe B \tye C \to A$ is an opfibration. Let $\alpha:\pi(a,b,c) \to a'$ be a morphism in $A$. We define the opcartesian lift of $\alpha$ to be
        \[ \alpha_! \defeq (\alpha, \id_{B(\alpha)b}, \id_{C(\alpha,\id_{B(\alpha)b})c}) : (a,b,c) \to (a',B(\alpha)b,C(\alpha,\id_{B(\alpha)b})c). \]
      \item That $\pi_2$ is a cartesian functor follows from the following equation.
        \[ \pi_2(\alpha, \id_B(\alpha)b, \id_{C(\alpha,\id_{B(\alpha)b})c})=(\alpha, \id_B(\alpha)b)\]
      \item For $e\defeq(a,b,c)$ in $A \ctxe B \tye C$, $\alpha:\pi_{A}(e) \to a'$ and $\beta: b' \to (\pi_2)_{|a}(e)$, we have
        \[ \alpha_! \beta^* e = (a',B(\alpha)b', C(\alpha,B(\alpha)\beta)c) = (B(\alpha)\beta)^*\alpha_!e \qedhere\]
    \end{enumerate}
  \end{proof}

  Next, we show the functorial action of unstraightening.

  \begin{proposition}
    Unstraightening of dependent profunctors (Definition \ref{def:unstraightening-dprof}) extends to a functor $\Phi:[A \cxe B^\op, \Cat] \to \DTSFib(A,B)$.
  \end{proposition}
  \begin{proof}
    Given dependent profunctors $C,D:A \ctxe B^\op \to \Cat$ and a natural transformation $\varphi:C \to D$, we now construct a functor $\Phi(\varphi):A \ctxe B \tye C \to A \ctxe B \tye D$.
    On objects, it maps $(a,b,c) : A \ctxe B \tye C$ to $(a,b,\varphi_{(a,b)}(c)) : A \ctxe B \tye D$.
    On morphisms, it maps $(\alpha,\beta,\theta):(a,b,c) \to (a',b',c')$ to $(\alpha,\beta,\varphi_{(a',B(\alpha)b)}\theta):(a,b,\varphi_{(a,b)}(c)) \to (a',b',\varphi_{(a,b)}(c'))$. Functoriality is immediate by naturality.
    We now show that $\Phi(\varphi)$ is a cartesian functor of D2SFib, i.e.\ that it satisfies the two conditions of \ref{def:d2sfib-1cells}
    \begin{enumerate}
      \item We show that $\Phi(\varphi)$ is a cartesian functor of opfibrations.
        Indeed, the opcartesian morphism
        \[ (\alpha, \id_{B(\alpha)b}, \id_{C(\alpha,\id_{B(\alpha)b})c}) : (a,b,c) \to (a',B(\alpha)b,C(\alpha,\id_{B(\alpha)b})c)\]
        in $A \ctxe B \tye C$ is mapped to
        \[ (\alpha, \id_{B(\alpha)b}, \id_{D(\alpha,\id_{B(\alpha)b})\Phi(\varphi)c}) : (a,b,\Phi(\varphi)c) \to (a',B(\alpha)b,D(\alpha,\id_{B(\alpha)b})\Phi(\varphi)c),\]
        which is also opcartesian.
      \item We show that for each $a \in A$, $\Phi(\varphi)_{|a}$ is a cartesian functor.
        Indeed, the cartesian morphism
        \[ (\id_a, \beta, \id_{C(\id_a,\beta)c}): (a,b,C(\id_a,\beta)c) \to (a,b',c)$ is mapped by $\Phi(\varphi) \]
        in $(A \ctxe B \tye C)_a$ is mapped to
        \[ (\id_a, \beta, \id_{D(\id_a,\beta)\Phi(\varphi)c}): (a,b,D(\id_a,\beta)\Phi(\varphi)c) \to (a,b',\Phi(\varphi)c), \]
        which is also cartesian.
    \end{enumerate}
    Functoriality of $\Phi$ is immediate from the definitions.
  \end{proof}

  \subsection{The straightening functor}
  We begin with a simple lemma.

  \begin{lemma}\label{d2sfib:lemma}
    Let $q:C \to A \ctxe B$ be a D2SFib, and let $\alpha:a \to a'$ in $A$ and $\beta: b \to b'$ in $B(a)$. Then the following diagram commutes.
    \[
      \begin{tikzcd}[column sep=4.5em]
        {C_{(a,b')}} & {C_{(a,b)}} \\
        {C_{(a',B\alpha b')}} & {C_{(a',B\alpha b)}}
        \arrow["{\beta^*}", from=1-1, to=1-2]
        \arrow["{\alpha_!}"', from=1-1, to=2-1]
        \arrow["{\alpha_!}", from=1-2, to=2-2]
        \arrow["{(B(\alpha)\beta)^*}"', from=2-1, to=2-2]
    \end{tikzcd}\]
  \end{lemma}
  \begin{proof}
    By definition of D2SFib, the diagram commutes on objects. To see that it commutes on morphisms, let $f:e \to d$ be a morphism in $C_{(a,b')}$. Then by either of the universal properties, there exists a unique diagonal arrow as below.
    \[
      \begin{tikzcd}[sep=2.25em]
        {\alpha_!\beta^*e} & {(B(\alpha)\beta)^*\alpha_!e} \\
        {\alpha_!\beta^*d} & {(B(\alpha)\beta)^*\alpha_!d}
        \arrow["{=}", no head, from=1-1, to=1-2]
        \arrow["{\alpha_!\beta^*f}"', from=1-1, to=2-1]
        \arrow[dashed, from=1-1, to=2-2]
        \arrow["{(B(\alpha)\beta)^*\alpha_!f}", from=1-2, to=2-2]
        \arrow["{=}"', no head, from=2-1, to=2-2]
      \end{tikzcd}
    \]
    This shows that $\alpha_!\beta^*f = (B(\alpha)\beta)^*\alpha_!f$, as desired.
  \end{proof}

  \begin{proposition}
    Let $q:C \to A \ctxe B$ be a D2SFib. Taking the fibre at each $(a,b) \in A \ctxe B$ produces a functor $\Psi(q):A \ctxe B^\op \to \Cat$, the \emph{straightening} of $q$.
  \end{proposition}
  \begin{proof}
    For $(a,b) \in A \ctxe B^\op$, we've defined $\Psi(C)(a,b) \defeq C_{(a,b)}$. Given $(\alpha,\beta) : (a,b) \to (a',b')$ in $A \ctxe B^\op$, we define $\Psi(C)(\alpha,\beta)$ by $C(a,b) \xrightarrow{\alpha_!} C(a',B(\alpha)b) \xrightarrow{\beta^*} C(a',b')$.

    We now check functoriality.
    First, $\Psi(C)$ preserves identities by splitness of the (op)fibrational data.
    It also preserves composition, since for $(a,b) \xrightarrow{(\alpha,\beta)} (a',b')  \xrightarrow{(\alpha',\beta')} (a'',b'')$, by Lemma \ref{d2sfib:lemma}, the following diagram commutes.
    \[
      \begin{tikzcd}
        C(a,b) & C(a',B(\alpha)b) & C(a'',B(\alpha'\alpha)b) \\
        & C(a',b') & C(a'',B(\alpha')b') & C(a'',b'')
        \arrow["\alpha_!", from=1-1, to=1-2]
        \arrow["\alpha'_!", from=1-2, to=1-3]
        \arrow["\beta^*"', from=1-2, to=2-2]
        \arrow["(B(\alpha)\beta')^*", from=1-3, to=2-3]
        \arrow["\alpha'_!"', from=2-2, to=2-3]
        \arrow["(\beta')^*"', from=2-3, to=2-4]
    \end{tikzcd}\]
  \end{proof}

  Next, we show the functorial action of straightening.

  \begin{proposition}
    Straightening of D2SFibs extends to a functor $\Psi:\DTSFib(A,B) \to [A \cxe B^\op, \Cat]$.
  \end{proposition}
  \begin{proof}
    Let $q : C \to A \ctxe B$ and $q' : D \to A \ctxe B$ be D2SFibs.
    Given a cartesian functor of D2SFibs $\varphi : C \to D$, we give a natural transformation $\Psi(\varphi):\Psi(C) \to \Psi(D)$. For $(a,b):A \ctxe B^\op$, the component $\Psi(\varphi)_{(a,b)}:C(a,b) \to D(a,b)$ is given by restricting $\varphi : C \to D$. It is natural since each square below commutes:
    \[
      \begin{tikzcd}
        C(a,b) & C(a',B(\alpha)b) & C(a',b') \\
        D(a,b) & D(a',B(\alpha)b) & D(a',b')
        \arrow["\alpha_!", from=1-1, to=1-2]
        \arrow["\varphi"', from=1-1, to=2-1]
        \arrow["\beta^*", from=1-2, to=1-3]
        \arrow["\varphi"{description}, from=1-2, to=2-2]
        \arrow["\varphi", from=1-3, to=2-3]
        \arrow["\alpha_!"', from=2-1, to=2-2]
        \arrow["\beta^*"', from=2-2, to=2-3]
    \end{tikzcd}\]
    by the corresponding theorem for (op)fibrations.
    Functoriality of $\Psi$ is immediate from the definitions.
  \end{proof}

  \subsection{The equivalence}

  \begin{proposition}\label{prop:d2sfibequiv}
    Straightening and unstraightening induce an equivalence of categories.
    \[ \Phi : [A \ctxe B^\op,\Cat] \simeq \mathsf{D2SFib}(A,B) : \Psi \]
  \end{proposition}
  \begin{proof}
    We begin by showing that $\Phi \circ \Psi \cong \id_{\mathsf{D2SFib}(A,B)}$. Let $q:C\to A \ctxe B$ be a D2SFib, we define a functor $\epsilon: C \to A \ctxe B \tye \Psi(q)$ as follows:
    \begin{itemize}
      \item \textbf{Objects:} For $c:C$, $\epsilon(c)\defeq(p(c), q_{|p(c)}(c), c)$.
      \item \textbf{Morphisms:} For $f:c\to c':C$, and setting $(a,b,c) \defeq \epsilon(c)$, $(a',b',c')\defeq \epsilon(c')$, $\alpha \defeq p(f):a \to a'$ and $\beta\defeq q_{|a'}(f):B(\alpha)b \to b'$, we define $\epsilon(f):\alpha_!c \to \beta^*c'$, by first obtaining a map $\alpha_!c \to c'$ by the universal property of $\alpha_!c$, and then obtaining $\epsilon(f)$ by using the universal property of $\beta^*c'$.
        \[
          \begin{tikzcd}
            {\alpha_!c} & {\beta^*c'} \\
            c & {c'}
            \arrow["\epsilon(f)", dashed, from=1-1, to=1-2]
            \arrow[dashed, from=1-1, to=1-2]
            \arrow[from=2-1, to=1-1]
            \arrow[dashed, from=1-1, to=2-2]
            \arrow[from=1-2, to=2-2]
            \arrow["f"', from=2-1, to=2-2]
        \end{tikzcd}\]
      \item \textbf{Functoriality:} Preservation of identities is immediate. For functoriality, let $c \xrightarrow{f} c' \xrightarrow{f'} c''$ be morphisms whose image under $q$ give the morphisms $(a,b) \xrightarrow{(\alpha,\beta)} (a',b') \xrightarrow{(\alpha',\beta')} (a'',b'')$. Then, preservation of composition follows from the diagram below.
        \[
          \begin{tikzcd}[column sep=1em,row sep=0.5em]
            & {(\alpha'\alpha)_!c} && {\alpha'_!\beta^*c'=(B(\alpha')\beta)^*\alpha'_!c'} && {(\beta' \circ B(\alpha')\beta)^*c''} \\
            \\
            {\alpha_!c} && {\beta^*c'} && {\alpha'_!c'} && {\beta'^*c''} \\
            \\
            & c && {c'} && {c''}
            \arrow["{\alpha'_!(Ff)}", from=1-2, to=1-4]
            \arrow["{(B\alpha'\beta)^*Ff'}", from=1-4, to=1-6]
            \arrow[from=1-4, to=3-5]
            \arrow[from=1-6, to=3-7]
            \arrow[from=3-1, to=1-2]
            \arrow["Ff", from=3-1, to=3-3]
            \arrow[from=3-3, to=1-4]
            \arrow[from=3-3, to=5-4]
            \arrow["{Ff'}", from=3-5, to=3-7]
            \arrow[from=3-7, to=5-6]
            \arrow[from=5-2, to=3-1]
            \arrow["f"', from=5-2, to=5-4]
            \arrow[from=5-4, to=3-5]
            \arrow["{f'}"', from=5-4, to=5-6]
        \end{tikzcd}\]
    \end{itemize}

    The inverse $\epsilon^{-1}:A \ctxe B \tye \Psi(q) \to C$ of $\epsilon$ is defined as follows.
    \begin{itemize}
      \item \textbf{Objects:} For $(a,b,c) : A \ctxe B \tye (\Psi q)$, we set $G(a,b,c) \defeq c$. \\
      \item \textbf{Morphisms:} For $(\alpha,\beta,\theta):(a,b,c) \to (a',b',c')$ in $A \ctxe B \tye (\Psi q)$, its image under $G$ is given by $c \to \alpha_! c \xrightarrow{\theta} \beta^* c' \to c'$
      \item \textbf{Functoriality:} Analogous to the proof of functoriality of $\epsilon$.
    \end{itemize}
    That these two functors are cartesian functors of D2SFibs is readily seen. Further, by construction, these are inverses.

    Finally, we show that $\Psi \circ \Phi \cong \id_{\mathsf{Functor}(A \ctxe B^\op,\Cat)}$.
    Let $C:A \ctxe B^\op \to \Cat$, we define a natural transformation $\eta:C \to \Psi(A \ctxe B \tye C)$, which is readily seen to be an isomorphism.
    \begin{itemize}
      \item \textbf{Components:} For $(a,b) : A \ctxe B^\op$, we define $\eta{(a,b)}:C(a,b) \to C(a,b)$ to be the identity functor.
      \item \textbf{Naturality:} For $(\alpha,\beta):(a,b) \to (a',b')$, naturality requires that $C(\alpha,\beta)=C(\id_{a'},\beta) \circ C(\alpha,\id_{B(\alpha)b})$, which holds by functoriality of $C$. \qedhere
    \end{itemize}
  \end{proof}

  We are now ready to prove Proposition \ref{prop:straight-d2sfibs}.
  \setcounter{proposition}{27}
  \begin{proposition}[Straightening-unstraightening for D2SFibs]
    Unstraightening of dependent profunctors (Definition \ref{def:unstraightening-dprof}) extends to a fibred equivalence of categories $\tye^* : \mathsf{DProf} \xrightarrow{\simeq} \DTSFib$ over $\ICat_{\textnormal{str}}$.
  \end{proposition}
  \begin{proof}
    This follows from the previous proposition, by noting that each construction given is natural in $A$ and $B$.
  \end{proof}

  \section{Expanded proofs of results}\label{sec:proofs}
  In this section, we present the proofs omitted in the main body of the paper.

  \subsection{The \texorpdfstring{$\Hom$}{Hom}-elimination rule}\label{subsec:Hom-elim}
  \setcounter{proposition}{31}
  \begin{proposition}[$\Hom$\textsc{-Elim}]
    The category model validates the following rule.
    {
      \small
      \begin{mathparpagebreakable}
        \inferrule*[Right=$\Hom$\textnormal{\textsc{-Elim}}]
        {
          \Gamma \vdash A\,\Ty \\\\
          \Gamma \cxe a:A \vdash X\,\Ty \\\\
          \Gamma \cxe b: A \cxe a:A \tscxe f : \Hom_A(\bar a,b) \cxe x : X^\op\vdash D\,\Ty \\\\
          \Gamma \cxe a:A \cxe x:X \vdash d \dvar D[\bar{a}/b,\Refl_A/f, \bar{x}/x]
        }
        {
          \Gamma \cxe b: A \cxe a:A \tscxe f : \Hom_A(\bar a,b) \cxe x : X \vdash j(f,x,d) \dvar D
        }
      \end{mathparpagebreakable}
    }%
    It also validates the following rule, which has the same antecedents as \textnormal{\textsc{$\Hom$-Elim}}, and has the following conclusion.
    \[ \Gamma \cxe a:A \cxe x:X \vdash j(\Refl_a,x,d) \jeq d \dvar D[\bar{a}/b,\Refl_A/f, \bar{x}/x] \hspace{3em} \textnormal{\textsc{$\Hom$-Comp}}\]
  \end{proposition}
  \begin{proof}
    We continue the sketched proof, using the introduced notation.
    We now give the action of $j$ on morphisms.
    Let $(\theta,(\alpha,\beta),\psi):(\gamma, f:a \to b,x) \to (\gamma',f':a'\to b', x')$ be a morphism in ${\Gamma \ctxe A^\to \ctxe X[\mathsf{dom}]}$.
    Note that this gives us a morphism $(\theta,\alpha,\psi): (\gamma, a,x) \to (\gamma', a',x')$ in the category $\Gamma \ctxe A \ctxe X$.
    By applying $d^*$ to it, we obtain the morphism
    \[ d^*(\theta,\alpha,\psi) : D(\theta,(\alpha,\alpha),\id_{X(\theta,\alpha)x})(\gamma,a,x) \to D(\id_{\gamma'},\id_{\id_{a'}},\psi)(\gamma',a',x') \]
    in $D(\gamma',\id_{a'},X(\theta,\alpha)x)$.

    We now define $j(\theta,(\alpha,\beta),\psi) := D(\id_{\gamma}',(\id_{a'},f'),\id_{X(\theta,\alpha)x})(d^*(\theta,\alpha,\psi))$.
    Commutativity of the diagram below ensures correctness of the endpoints.
    {
      \small
      \[
        \begin{tikzcd}[row sep=2em, column sep=8em]
          {D(\gamma,\id_a,x)} & {D(\gamma',\id_{a'},X(\theta,\alpha)x)} & {D(\gamma',\id_{a'},x')} \\
          {D(\gamma,f,x)} & {D(\gamma',f',X(\theta,\alpha)x)} & {D(\gamma',f',x')}
          \arrow["{D(\theta,(\alpha,\alpha),\id_{X(\theta,\alpha)x})}", from=1-1, to=1-2]
          \arrow["{D(\id_\gamma,(\id_a,f),\id_x)}"{description}, from=1-1, to=2-1]
          \arrow["{D(\id_{\gamma}',(\id_{a'},f'),\id_{X(\theta,\alpha)x})}"{description}, from=1-2, to=2-2]
          \arrow["{D(\id_{\gamma'},\id_{\id_{a'}},\psi)}"', from=1-3, to=1-2]
          \arrow["{D(\id_{\gamma}',(\id_{a'},f'),\id_{x'})}"{description}, from=1-3, to=2-3]
          \arrow["{D(\theta,(\alpha,\beta),\id_{X(\theta,\alpha)x})}"', from=2-1, to=2-2]
          \arrow["{D(\id_{\gamma'},\id_{f'},\psi)}", from=2-3, to=2-2]
      \end{tikzcd}\]
    }
    Functoriality readily follows and, by construction, the desired square commutes.
  \end{proof}

  \subsection{A proof of Yoneda lemma}\label{subsec:yoneda-full}
  In this section we give a complete proof of Yoneda lemma.
  We begin with a technical lemma.
  \setcounter{proposition}{59}
  \begin{lemma}\label{lemma:yoneda-perm-proof}
    Let $A : \UU$ be a type in the empty context.
    Then, for a context $\Gamma$ of the form $\Gamma \defeq \gamma_1 : \Gamma_1 \cxe \cdots \cxe \gamma_n : \Gamma_n$, we have an isomorphism of contexts
    \[ \Gamma \cxe x:A \cxe a:A \tscxe \Hom(\bar a, x) \cong \Gamma \cxe a:A \cxe^- (x,b):\sum_{x:A}\Hom(a,x).\]
  \end{lemma}
  \begin{proof}
    The lemma follows from the following chain of isomorphisms of contexts, all of which follow from the rules for $\Sigma$, opposite contexts, and opposite types.
    {
      \allowdisplaybreaks
      \begin{align*}
        &\Gamma \cxe x:A \cxe a:A \tscxe \Hom(\bar a, x) \\
        & \hspace{2em} \cong
        \Gamma \cxe x:A \cxe z : \sum^\tscxe_{a:A}\Hom(\bar a,x) \\
        & \hspace{2em} \cong
        \Gamma \cxe z : \sum_{x:A} \sum^\tscxe_{a:A}\Hom(\bar a,x)  \\
        & \hspace{2em} \cong
        \Gamma \cxe^- z : \sum_{x:A} \sum^\tscxe_{a:A}\Hom(\bar a,x)  \\
        & \hspace{2em} \jeq
        \left(\Gamma^\op \cxe z : \left(\sum_{x:A} \sum^\tscxe_{a:A}\Hom(\bar a,x) \right)^\op \right)^\op\\
        & \hspace{2em} \cong
        \left(\Gamma^\op \cxe z : \sum_{a:A^\op} \sum^\tscxe_{x:A^\op}\Hom(a,\bar x)^\op  \right)^\op\\
        & \hspace{2em} \jeq
        \left(\Gamma^\op \cxe z : \sum_{a:A^\op} \sum^\tscxe_{x:A^\op}\Hom(a,\bar x)  \right)^\op\\
        & \hspace{2em} \cong
        \left(\Gamma^\op \cxe a:A^\op \cxe z : \sum_{x:A^\op}^\tscxe \Hom(a,\bar x) \right)^\op\\
        & \hspace{2em} \jeq
        \left(\Gamma^\op \cxe a:A^\op \cxe z : \left(\sum_{x:A} \Hom(a,\bar x)^\op\right)^\op \right)^\op\\
        & \hspace{2em} \jeq
        \left(w : (\Sigma\Gamma \times A)^\op \cxe z : \left(\sum_{x:A}\Hom(a,x)\right)^\op \right)^\op\\
        & \hspace{2em} \jeq
        w : \Sigma\Gamma \times A \cxe^- z : \sum_{x:A}\Hom(a,x)\\
        & \hspace{2em} \cong
        \Gamma \cxe a : A \cxe^- z : \sum_{x:A}\Hom(a,x)
      \end{align*}
    }%
    Here, $\Sigma\Gamma$ is defined by $\Sigma\Gamma \defeq \Sigma_{\gamma_1 : \Gamma_1}\cdots \Sigma_{\gamma_{n-1}:\Gamma_{n-1}}\Gamma_n$.
  \end{proof}

  \setcounter{proposition}{49}
  \begin{lemma}[Yoneda]
    Let $A:\UU$ be a type in the empty context. Then the two following functors are naturally isomorphic.
    \begin{align*}
      Y, \bar Y   & :(A \to \Set) \times A \to \Set             \\
      Y(F,a)      & \defeq \Nat(\Hom_A(a,-), F) \\
      \bar Y(F,a) & \defeq Fa
    \end{align*}
  \end{lemma}
  \begin{proof}
    We use Corollary \ref{cor:nat-isos}.
    First, we obtain a natural transformation $\Phi:Y \to \bar Y$ as follows.
    {
      \small
      \begin{mathparpagebreakable}
        \inferrule*[Right=$\Pi$-Elim]
        {
          \inferrule*[Right=Weak]
          {
            \inferrule*[Right=Def]
            {
              \inferrule*[Right=$\Sigma^\op$-Intro]
              {F:(A \to \Set)^\op \cxe a :A^\op  \vdash a:A^\op \\\\
              F:(A \to \Set)^\op \cxe a :A^\op  \vdash \Refl_a \dvar \hom(a,a)}
              {
                (F:A \to \Set \cxe a :A)^\op  \vdash (a,\Refl_a) : \bigg(\sum_{x:A}\hom(a,x)\bigg)^\op
              }
            }
            {
              F:A \to \Set \cxe a :A  \vdash (a,\Refl_a) \cvar \sum_{x:A}\hom(a,x)
            }
          }
          {
            F:A \to \Set \cxe a :A \cxe \varphi: Y(F,a) \vdash (a,\Refl_a) \cvar \sum_{x:A}\hom(a,x)
          }
        }
        {
          F:A \to \Set \cxe a :A \cxe \varphi: Y(F,a)\vdash \varphi(a,\Refl_a) : Fa
        }
      \end{mathparpagebreakable}
    } %
    Its inverse $\Psi:\bar Y \to Y$ is given by $\Hom$-elimination.
    {
      \small
      \begin{mathpar}
        \inferrule*[Right=$\Pi$-Intro]
        {
          \inferrule*[Right=Perm]
          {
            \inferrule*[Right=Lemma-\ref{lemma:yoneda-perm-proof}]
            {
              \inferrule*[Right=$\Hom$-Elim]
              {
                F:A \to \Set \cxe a:A \vdash Fa\,\Ty \\\\
                F:A \to \Set \cxe x:A \cxe a:A \tscxe f : \hom(a,x), r:Fa \vdash Fx\,\Ty \\\\
                F:A \to \Set \cxe a :A \cxe r:Fa \vdash r : Fa
              }
              {
                F:A \to \Set \cxe x:A \cxe a:A \tscxe f : \hom(a,x) \cxe r:Fa \vdash j(f,r,r) : Fx
              }
            }
            {
              F:A \to \Set \cxe a :A \cxe^- (x,f) : \sum_{x:A}\hom(a,x)\cxe r : Fa \vdash j(f,r,r) : Fx
            }
          }
          {
            F:A \to \Set \cxe a :A \cxe r : Fa \cxe^- (x,f) : \sum_{x:A}\hom(a,x)\vdash j(f,r,r) : Fx
          }
        }
        {
          F:A \to \Set \cxe a :A \cxe r:Fa \vdash \lambda (x,f) . j(f,r,r) : \prod_{(x,f):\sum_{x:A}\hom(a,x)} Fx
        }
      \end{mathpar}
    }

    Now we show that these are inverses. That $\Phi \circ \Psi$ is the identity follows from the computation rule for $\Hom$-types.
    The other direction, that $\Psi \circ \Phi = \id$, is derived below, using the abbreviations $j_{\varphi} \defeq j(f,\varphi(a,\Refl_a),\varphi(a,\Refl_a))$ and $j_{\Refl} \defeq j(f, \varphi, \Refl_{\varphi(a,\Refl_a)})$.
    {
      \small
      \begin{mathpar}
        \inferrule*[right=dfunext]
        {
          \inferrule*[Right=Perm]
          {
            \inferrule*[Right=Lemma-\ref{lemma:yoneda-perm-proof}]
            {
              \inferrule*[Right=$\Hom$-Elim]
              {
                F:A \to \Set\cxe a:A \vdash Y(F,a)\,\Ty \\\\
                F:A \to \Set\cxe x :A\cxe a:A \tscxe f : \Hom(\bar a,x)\cxe  \varphi: Y(F,a) \vdash \Hom_{Fx}(\varphi(x,f), j_{\varphi})\,\Ty \\\\
                F:A \to \Set\cxe a :A\cxe \varphi: Y(F,a) \vdash \Refl_{\varphi(a,\Refl_a)} : \Hom_{Fx}(\varphi(a,\Refl_a), \varphi(a,\Refl_a))
              }
              {
                F:A \to \Set\cxe x :A\cxe a:A \tscxe f : \Hom(\bar a,x) \cxe \varphi: Y(F,a) \vdash j_{\Refl} : \Hom_{Fx}(\varphi(x,f), j_{\varphi})
              }
            }
            {
              F:A \to \Set\cxe a :A\cxe^- (x,f): \sum_{x:A}\Hom(a,x) \cxe \varphi: Y(F,a) \vdash j_{\Refl} : \Hom_{Fx}(\varphi(x,f), j_{\varphi})
            }
          }
          {
            F:A \to \Set\cxe a :A\cxe \varphi: Y(F,a) \cxe^- (x,f): \sum_{x:A}\Hom(a,x) \vdash j_{\Refl} : \Hom_{Fx}(\varphi(x,f), j_{\varphi})
          }
        }
        {
          F:A \to \Set \cxe a :A \cxe \varphi: Y(F,a) \vdash \mathsf{dfunext}(j_{\Refl_{\varphi}}) : \Hom_{Y(F,a)}(\varphi, \lambda(x,f).j_{\varphi})
        }
      \end{mathpar}
    }
  \end{proof}

  \end{document}